\documentclass[a4paper,USenglish,cleveref, autoref, thm-restate]{oasics-v2021}

\title{A Simpler Approach for Monotone Parametric Minimum Cut: Finding the Breakpoints in Order}

\titlerunning{A Simpler Approach for Monotone Parametric Minimum Cut}

\newcommand{\affiliationBonn}{University of Bonn, Germany}

\author{Arne Beines}{Formerly of \affiliationBonn}{beines-arne@t-online.de}{}{}

\author{Michael Kaibel}{\affiliationBonn}{s6mikaib@uni-bonn.de}{https://orcid.org/0009-0006-1967-5376}{}

\author{Philip Mayer}{\affiliationBonn}{pmayer@cs.uni-bonn.de}{https://orcid.org/0009-0007-4800-7753}{}

\author{Petra Mutzel}{\affiliationBonn}{pmutzel@cs.uni-bonn.de}{https://orcid.org/0000-0001-7621-971X}{}

\author{Jonas Sauer}{\affiliationBonn}{jsauer1@uni-bonn.de}{https://orcid.org/0000-0002-7196-7468}{}

\authorrunning{A. Beines and M. Kaibel and P. Mayer and P. Mutzel and J. Sauer}

\Copyright{Arne Beines and Michael Kaibel and Philip Mayer and Petra Mutzel and Jonas Sauer}



\supplement{The benchmark instances for polygon aggregation are available at~\url{https://doi.org/10.5281/zenodo.13642985}. Source code is available at~\url{https://doi.org/10.5281/zenodo.13788773}.}

\funding{This work was partially funded by the Deutsche Forschungsgemeinschaft (DFG, German Research Foundation) under grant FOR-5361 -- 459420781, and by the BMBF (Germany) and state of NRW as part of the Lamarr-Institute, LAMARR22B.}

\nolinenumbers
\hideOASIcs

\usepackage{packages/algorithm}
\usepackage{packages/colors}
\usepackage{packages/figures}
\usepackage{booktabs}
\usepackage{mathtools}

\DeclareMathOperator*{\argmin}{arg\,min}

\newcommand{\MPMC}{MPMC\xspace}
\newcommand{\PBFS}{PBFS\xspace}
\newcommand{\dichotomicscheme}{DS\xspace}

\newcommand{\graph}{G}
\newcommand{\residualGraph}{\graph_R}
\newcommand{\vertices}{V}
\newcommand{\newVertices}{\vertices_N}
\newcommand{\edges}{E}
\newcommand{\saturatedEdges}{\edges_\text{sat}}
\newcommand{\revEdges}{\overleftarrow{\edges}}
\newcommand{\revEdgesSat}{\overleftarrow{\edges}_\text{sat}}
\newcommand{\outEdges}{\hat{\edges}_\text{out}}
\newcommand{\residualEdges}{\edges_R}
\newcommand{\vertex}{v}
\newcommand{\vertexA}{u}
\newcommand{\vertexB}{v}
\newcommand{\vertexC}{w}
\newcommand{\source}{s}
\newcommand{\sink}{t}

\newcommand{\contractedGraph}{\hat{\graph}}
\newcommand{\contractedVertices}{\hat{\vertices}}
\newcommand{\contractedEdges}{\hat{\edges}}
\newcommand{\residualContractedGraph}{\contractedGraph_R}
\newcommand{\residualContractedEdges}{\contractedEdges_R}

\newcommand{\edge}{e}
\newcommand{\tree}{T}
\newcommand{\treeVertices}{\vertices_\sink}
\newcommand{\treeEdges}{\edges_\tree}
\newcommand{\capacity}{c}
\newcommand{\residualCapacity}{c_R}
\newcommand{\flow}{f}
\newcommand{\newFlow}{\flow_\text{new}}
\newcommand{\initialFlow}{\flow_\text{init}}
\newcommand{\netInflow}{\flow_\text{net}}
\DeclareMathOperator{\excess}{ex}
\newcommand{\cut}{C}
\newcommand{\breakpoint}{\beta}
\newcommand{\orphans}{\vertices_O}
\newcommand{\excessVertices}{\vertices_\text{ex}}
\newcommand{\minLambda}{\lambda_\text{min}}
\newcommand{\maxLambda}{\lambda_\text{max}}
\newcommand{\rootLambda}{\lambda_\text{root}}
\newcommand{\distanceLabel}{d_\sink}
\DeclareMathOperator{\dist}{dist}

\DeclareMathOperator{\children}{children}

\SetKwFunction{Bisect}{Bisect}
\SetKwFunction{Initialize}{Initialize}
\SetKwFunction{CalcMaxFlow}{CalcMaxFlow}
\SetKwFunction{BuildResidualGraph}{BuildResidualGraph}
\SetKwFunction{ReverseBFS}{ReverseBFS}
\SetKwFunction{CalcFlowFunction}{CalcFlowFunction}
\SetKwFunction{CalcFlowFunctionAlt}{Calc-FlowFunction}
\SetKwFunction{CalcRoot}{CalcRoot}
\SetKwFunction{SetFlow}{SetFlow}
\SetKwFunction{UpdateResidualGraph}{UpdateResidualGraph}
\SetKwFunction{RemoveTreeEdge}{RemoveTreeEdge}
\SetKwFunction{ReconnectTree}{ReconnectTree}
\SetKwFunction{DrainExcess}{DrainExcess}
\SetKwFunction{Pop}{Pop}
\SetKwFunction{AdoptWithSameDist}{AdoptSameDist}
\SetKwFunction{AdoptWithNewDist}{AdoptNewDist}

\newcommand{\polygons}{\mathcal P}
\newcommand{\triangulation}{\mathcal D}
\DeclareMathOperator*{\convexHull}{conv}
\newcommand{\solution}{S}
\newcommand{\solutions}{\mathcal S}
\newcommand{\area}{A}
\newcommand{\perimeter}{P}
\newcommand{\triangles}{T}

\begin{document}

\maketitle

\begin{abstract}
We present parametric breadth-first search~(\PBFS), a new algorithm for solving the parametric minimum cut problem in a network with source-sink-monotone capacities.
The objective is to find the set of breakpoints, i.e., the points at which the minimum cut changes.
It is well known that this problem can be solved in the same asymptotic runtime as the static minimum cut problem.
However, existing algorithms that achieve this runtime bound involve fairly complicated steps that are inefficient in practice.
\PBFS uses a simpler approach that discovers the breakpoints in ascending order, which allows it to achieve the desired runtime bound while still performing well in practice.
We evaluate our algorithm on benchmark instances from polygon aggregation and computer vision.
Polygon aggregation was recently proposed as an application for parametric minimum cut, but the monotonicity property has not been exploited fully.
\PBFS outperforms the state of the art on most benchmark instances, usually by a factor of~2--3.
It is particularly strong on instances with many breakpoints, which is the case for polygon aggregation.
Compared to the existing min-cut-based approach for polygon aggregation, \PBFS scales much better with the instance size.
On large instances with millions of vertices, it is able to compute all breakpoints in a matter of seconds.
\end{abstract}

\newpage

\section{Introduction}
Many optimization problems can be transformed into finding a minimum cut between a source and sink vertex in a graph. Some examples are image segmentation \cite{BJ01,BF06,BK03}, polygon aggregation for map simplification \cite{RDGRH21} and community detection in social networks \cite{CaiSHYH05}. Finding a minimum cut is equivalent to the problem of finding a maximum flow, which has been extensively studied from a theoretical as well as practical perspective.

In some applications, the edge capacities are not static but functions of some parameter~$\lambda$. As~$\lambda$ changes, so does the minimum cut. In this parametric version of the problem, the objective is to compute a minimum cut for each possible value of~$\lambda$.
An important subclass of this is the \emph{source-sink monotone parametric minimum cut~(\MPMC) problem}~\cite{GGT89}. Here, edge capacities are non-decreasing in~$\lambda$ for source-incident edges, non-increasing for sink-incident edges, and constant for all others. Solutions for the \MPMC problem are \emph{nested}: the source component of a minimum cut for~$\lambda$ is fully contained in the source component of a minimum cut for all~$\lambda'>\lambda$. The values of~$\lambda$ for which the minimum cut changes are called \emph{breakpoints}. A consequence of the nesting property is that the number of breakpoints is linear in the number of vertices in the graph.

Several applications are source-sink-monotone~\cite{GGT89,KBR07}.
Notable examples are binary image segmentation, where~$\lambda$ controls the penalty for separating similar neighboring pixels~\cite{KBR07}, and polygon aggregation for map simplification, where~$\lambda$ represents the desired zoom factor.
For the latter problem, Rottmann et al.~\cite{RDGRH21} give a transformation to a parametric min-cut problem that is not monotone but has the nesting property.
We show that it can be reformulated as a \MPMC problem, which enables the use of state-of-the-art \MPMC algorithms to solve the problem faster.
Additionally, with the structure of the aggregation instances in mind, we propose a new algorithm for the \MPMC problem.

\subparagraph*{State of the Art.}
If the breakpoints are supplied in advance, the \MPMC problem can be solved with a simple adaptation of the push-relabel~(PRF) algorithm for max-flows~\cite{GT88}.
This is done by exploiting the \emph{flow update property}: a max-flow for a parameter value~$\lambda_i$ can be transformed into a starting point for the flow computation of the next parameter value~$\lambda_{i+1}>\lambda_i$~\cite{GGT89}.
As a consequence, all minimum cuts can be computed in the same asymptotic runtime as a single call to PRF.
Practical implementations of PRF run in~$\mathcal O(n^2\sqrt{m})$ time on a graph with~$n$ vertices and~$m$ edges~\cite{CM88}.

For the case where the breakpoints are not known in advance, the GGT algorithm~\cite{GGT89} achieves the same asymptotic runtime bound.
However, it is significantly more complicated and has several drawbacks in practice.
To find the breakpoints, the algorithm recursively bisects the interval of possible~$\lambda$ values.
In each inspected interval, parts of the graph that are already known to lie in the source or sink component are contracted.
In practice, this step requires a significant portion of the overall runtime.
The flow update property is exploited with a fairly complicated scheme that performs two calls to the flow algorithm simultaneously in each interval and only keeps the results of the faster one.
In practice, a simplified version of the algorithm with flow updates is faster~\cite{BDGTZ07}.
However, it only achieves a worst-case runtime of~$\mathcal O(n^3\sqrt{m})$.

\subparagraph*{Our Contribution.}
We propose \emph{parametric breadth-first search} (\PBFS), a new algorithm for the \MPMC problem.
It incorporates ideas from the \emph{incremental breadth-first search} (IBFS)~\cite{GHKTW11} algorithm for the static max-flow problem, which has a runtime bound of~$\mathcal O(n^2m)$ and is known to perform well on (non-parametric) image segmentation instances.
\PBFS matches this bound, which means that, like GGT, it solves the parametric problem in the same asymptotic time as computing a single static min-cut.

Unlike GGT, \PBFS discovers the breakpoints in ascending order.
This allows flow updates to be incorporated in a straightforward and practically efficient manner.
\PBFS works by maintaining a shortest-path tree in the sink component of the residual graph.
As~$\lambda$ increases, the current max-flow is preserved by sending any excess flow towards the sink through the tree.
The resulting gradual changes to the flow are tracked by maintaining flow functions for the edges.
This way, the algorithm detects when a tree edge becomes saturated, which changes the residual graph and triggers a repair operation for the tree.
If the sink component becomes disconnected, a new breakpoint has been found.

We evaluate \PBFS on a set of benchmark instances from polygon aggregation and image segmentation.
We show that the performance benefits of IBFS carry over to aggregation instances.
On most instances, including all aggregation instances, \PBFS outperforms the state of the art, usually by a factor of~2--3.
It is slower by a similar margin on some segmentation instances that have few breakpoints and large structural changes in the residual graph between the breakpoints.

Overall, \PBFS is the first algorithm for the \MPMC problem to achieve good practical performance while bounding the worst-case runtime to that of a single flow computation.
Additionally, it drastically improves upon the runtime of the aggregation approach by Rottmann et al.~\cite{RDGRH21}, which has to resort to approximation even on moderately sized instances.
By contrast, \PBFS computes the exact breakpoint function in under a minute on instances with millions of vertices.

\section{Preliminaries}
\label{sec:preliminaries}
\subsection{Constant Capacities}
A \emph{(constant) flow network} is a graph~$\graph=(\vertices,\edges,\capacity,\source,\sink)$ consisting of a set~$\vertices$ of~$n$ vertices, a set~$\edges\subseteq\vertices\times\vertices$ of~$m$ directed edges, a function~$\capacity\colon \edges \to \mathbb{R}_{\geq 0}$ that assigns a \emph{capacity}~$\capacity(\edge)$ to each edge~$\edge$, and designated source and sink vertices~$\source,\sink\in\vertices$.
For simplicity, we assume that for every edge~$(\vertexA,\vertexB)\in\edges$, the graph also contains the reverse edge~$(\vertexB,\vertexA)$.
If necessary, $(\vertexB,\vertexA)$ can be added to~$\edges$ with capacity~$0$.
The capacity of an edge set~$\edges'\subseteq\edges$ is the sum of the individual edge capacities, i.e., $\capacity(\edges'):=\sum_{\edge\in\edges'}\capacity(\edge)$.

\subparagraph*{Flows.}
An $(\source, \sink)$-\emph{flow} in~$\graph$ is a function~$\flow\colon \edges \to \mathbb{R}$ that fulfills the following constraints:
\begin{itemize}
    \item \textbf{Capacity:} $\flow(\vertexA,\vertexB) \leq \capacity(\vertexA,\vertexB)$ for $(\vertexA,\vertexB) \in \edges$
    \item \textbf{Antisymmetry:} $\flow(\vertexA,\vertexB)=-\flow(\vertexB,\vertexA)$ for~$(\vertexA,\vertexB)\in\edges$
    \item \textbf{Flow conservation:} $\sum_{(\vertexB,\vertexA) \in \edges} \flow(\vertexB,\vertexA) = 0$ for $\vertexA \in \vertices\setminus\{\source,\sink\}$
\end{itemize}
The \emph{value} of~$\flow$ is given by~$\sum_{(\vertexB,\sink) \in \edges} \flow(\vertexB,\sink)$.
We call $\flow$ a \emph{maximum flow} if it has maximal value among all~$(\source, \sink)$-flows. 
The flow over an edge set~$\edges'\subseteq\edges$ is the sum of the individual edge flows, i.e., $\flow(\edges'):=\sum_{\edge\in\edges'}\flow(\edge)$.
If~$\flow$ fulfills the capacity and antisymmetry constraints but (possibly) violates the flow conservation constraints, it is called a~\emph{pseudoflow}.
The \emph{excess} of a vertex~$\vertexA$ is defined as~$\excess(\vertexA):=\sum_{(\vertexB,\vertexA) \in \edges} \flow(\vertexB,\vertexA)$.
Thus, the flow conservation constraint for~$\vertexA$ is fulfilled if~$\excess(\vertexA)=0$.
If all excesses are non-negative, $\flow$ is called a~\emph{preflow}.

The \emph{residual capacity} of an edge~$\edge=(\vertexA,\vertexB) \in \edges$ is defined as~$\residualCapacity(\edge) := \capacity(\edge) - \flow(\edge)$.
The capacity constraint for~$\edge$ is fulfilled if~$\residualCapacity(\edge) \geq 0$.
We call~$\edge$ \emph{saturated} if~$\residualCapacity(\edge)=0$ and \emph{residual} if~$\residualCapacity(\edge)>0$.
The \emph{residual graph}~$\residualGraph := (\vertices, \residualEdges,\residualCapacity,\source,\sink)$ contains only the residual edges, which are denoted by~$\residualEdges\subseteq\edges$.

\subparagraph*{Cuts.}
An~$(\source, \sink)$-\emph{cut} in~$\graph$ is a partition~$\cut=(\vertices_\source,\vertices_\sink)$ of~$\vertices$ into a \emph{source component}~$\vertices_\source$ with~$\source\in\vertices_\source$ and a \emph{sink component}~$\vertices_\sink$ with~$\sink\in\vertices_\sink$.
The capacity~$|\cut|$ of the cut is the capacity of the \emph{cut set}~$\edges(\cut)$, which is the set of edges~$(\vertexA,\vertexB)$ with~$\vertexA \in \vertices_\source$ and~$\vertexB \in \vertices_\sink$.
We call~$\cut$ a \emph{minimum cut} if it has minimal value among all~$(\source, \sink)$-cuts, and a \emph{sink-minimal minimum cut} if~$|\vertices_\sink|$ is minimal among all minimum cuts.

\subparagraph*{Relationship Between Flows and Cuts.}
Given a flow network~$\graph$, a flow~$\flow$ in~$\graph$ and a cut~$\cut=(\vertices_\source,\vertices_\sink)$ in~$\graph$, the flow of~$\flow$ over~$\cut$ is defined as the flow~$\flow(\cut)=\flow(\edges(\cut))$ over the cut set.
It follows from flow conservation that this is equal to the value of~$\flow$.
We say that~$\flow$ \emph{saturates}~$\cut$ if it saturates all edges in~$\edges(\cut)$.
The \emph{max-flow-min-cut theorem}~\cite{FF56} states that the value of a maximum flow is equal to the capacity of a minimum cut.
The following is a corollary of this:

\begin{observation} \label{obs:saturateAllCuts}
    A maximum flow~$\flow$ in a flow network~$\graph$ saturates every minimum cut~$\cut$ in~$\graph$.
\end{observation}
\begin{proof}
    The flow over~$\cut$ is equal to the value of~$\flow$, which is equal to the capacity of~$\cut$.
\end{proof}

Given a flow network~$\graph$ and a maximum flow~$\flow$ in~$\graph$ with residual graph~$\residualGraph$, we define~$\vertices_\sink(\flow)$ as the set of vertices from which~$\sink$ is reachable in~$\residualGraph$.
The corresponding cut is given by $\cut(\flow):=(\vertices_\source(\flow),\vertices_\sink(\flow))$ with~$\vertices_\source(\flow)=\vertices\setminus\vertices_\sink(\flow)$.

\begin{corollary} \label{cor:sinkMinimalCut}
    For any flow network~$\graph$, there is a unique sink-minimal minimum cut~$\cut=(\vertices_\source,\vertices_\sink)$, which equals~$\cut(\flow)$ for every maximum flow~$\flow$.
\end{corollary}
\begin{proof}
    Consider a maximum flow~$\flow$ and a sink-minimal minimum cut~$\cut=(\vertices_\source,\vertices_\sink)$.
    By Observation~\ref{obs:saturateAllCuts}, $\flow$ saturates~$\cut$.
    Thus, for any vertex~$\vertex\in\vertices_\source$, there is no residual~$\vertex$-$\sink$-path in~$\graph$, which implies~$\vertices_\source\subseteq\vertices_\source(\flow)$ and thereby~$\vertices_\sink\supseteq\vertices_\sink(\flow)$.
    Because~$|\vertices_\sink|$ is minimal among all minimum cuts, equality must hold.
\end{proof}

\subparagraph*{Contraction.}
For a vertex set~$\vertices' \subseteq \vertices \setminus \{ \sink \}$, the \emph{source-contracted graph}~$\graph/\vertices'$ is obtained from~$\graph$ by contracting all vertices in~$\vertices'$ into~$\source$.
A vertex~$\vertexA$ is \emph{contracted} into~$\source$ by removing it from the graph and replacing all incident edges~$(\vertexA,\vertexB)$ or~$(\vertexB,\vertexA)$ with edges~$(\source,\vertexB)$ or~$(\vertexB,\source)$, respectively.
Any resulting multi-edges are replaced with a single edge whose capacity is the sum of the individual edge capacities.
The resulting source-contracted edge set is denoted as~$\edges/\vertices'$.

\subparagraph*{Distance Labels.}
A \emph{distance labeling} for the residual graph~$\residualGraph := (\vertices, \residualEdges)$ 
is a function~$\distanceLabel\colon \vertices \to \mathbb{N}_{0}$ with~$\distanceLabel(\sink) = 0$.
We call the labeling \emph{valid} if~$\distanceLabel(\vertexA) \leq \distanceLabel(\vertexB)+1$ holds for all edges $\edge = (\vertexA,\vertexB) \in \residualEdges$.
In a valid distance labeling, $\distanceLabel(\vertex)$ is a lower bound for the length of the shortest~$\vertex$-$\sink$-path.
An edge $\edge = (\vertexA,\vertexB) \in \residualEdges$ is called \emph{admissible} with respect to~$\distanceLabel$ if~$\distanceLabel(\vertexA) = \distanceLabel(\vertexB)+1$.

\subsection{Parametric Capacities}
A \emph{parametric flow network}~$\graph=(\vertices,\edges,\capacity,[\minLambda,\maxLambda],\source,\sink)$ differs from a constant flow network in that the capacity of each edge~$\edge$ is a function~$\capacity(\edge)\colon[\minLambda,\maxLambda]\to\mathbb{R}_{\geq 0}$.
In all applications considered in this paper, the capacity can be modelled as an affine function~$\capacity(\edge)(\lambda)=c_1 \cdot \lambda + c_2$.
However, the presented algorithms support any type of function for which the roots can be calculated efficiently.
We say that~$\capacity$ (and thereby~$\graph$) is \emph{(source-sink-)monotone} if
\begin{itemize}
    \item $\capacity(\source,\vertex)$ is non-decreasing for all $\vertex \in \vertices\setminus\{\sink\}$,
    \item $\capacity(\vertex,\sink)$ is non-increasing for all $\vertex \in \vertices\setminus\{\source\}$, and
    \item $\capacity(\vertexA,\vertexB)$ is constant for all $\vertexA \in \vertices\setminus\{\source\},\vertexB \in \vertices\setminus\{\sink\}$.
\end{itemize}

For a parameter~$\lambda \in [\minLambda,\maxLambda]$, we denote by~$\capacity[\lambda]\colon \edges \to \mathbb R_{\geq 0}$ the constant edge capacities that are obtained by evaluating the capacity function of each edge~$\edge$ at~$\lambda$, i.e., $\capacity[\lambda](\edge) := \capacity(\edge)(\lambda)$.
The corresponding constant flow network is denoted by~$\graph[\lambda]:=(\vertices,\edges,\capacity[\lambda],\source,\sink)$.
A \emph{flow function} assigns to each edge~$\edge$ a function~$\flow(\edge)\colon[\minLambda,\maxLambda]\to\mathbb{R}$ such that~$\flow(\edge)(\lambda)$ fulfills the antisymmetry and flow conservation constraints for every parameter values~$\lambda\in[\minLambda,\maxLambda]$, and there is at least one parameter value~$\lambda'$ for which~$\flow(\edge)(\lambda')$ fulfills the capacity constraint.
As with the capacities, we use the notation~$\flow[\lambda](\edge):=\flow(\edge)(\lambda)$.
Note that~$\flow[\lambda']$ is a flow, but~$\flow[\lambda]$ for~$\lambda\neq\lambda'$ may not be.

The minimum cut problem can be extended to the \emph{parametric minimum cut problem}:
Given a parametric flow network~$\graph=(\vertices,\edges,\capacity,[\minLambda,\maxLambda],\source,\sink)$, compute a set~$\mathcal{C}$ of cuts that contains a minimum cut in~$\graph[\lambda]$ for every parameter value~$\lambda\in[\minLambda,\maxLambda]$.
If there are no restrictions on the capacity functions, the size of~$\mathcal{C}$ cannot be bounded.
However, if the network is source-sink-monotone, the~\emph{nesting property}~\cite{GGT89} holds:
\begin{observation}
\label{lem:nesting}
Let~$\graph=(\vertices,\edges,\capacity,[\minLambda,\maxLambda],\source,\sink)$ be a monotone parametric flow network and $\lambda_1,\lambda_2\in[\minLambda,\maxLambda]$ two parameter values with~$\lambda_1<\lambda_2$.
Further, let~$(\vertices^1_\source,\vertices^1_\sink)$ and~$(\vertices^2_\source,\vertices^2_\sink)$ be the sink-minimal minimum cuts in~$\graph[\lambda_1]$ and~$\graph[\lambda_2]$, respectively.
Then~$\vertices^1_\source \subseteq \vertices^2_\source$ and~$\vertices^1_\sink \supseteq \vertices^2_\sink$.
\end{observation}
This property implies that there is a solution set of size~$\mathcal{O}(n)$.
To represent it in a compact manner, we define the \emph{breakpoint function}.
This is a function~$\breakpoint\colon \vertices \to \mathbb{R}_{\geq 0}$ with the following property: for each parameter value~$\lambda \in [\minLambda,\maxLambda]$, the sink-minimal minimum cut in~$\graph[\lambda]$ is given by
\begin{align*}
\cut(\breakpoint,\lambda) &:= (\vertices_\source(\breakpoint,\lambda),\vertices_\sink(\breakpoint,\lambda)) \text{ with}\\
\vertices_\source(\breakpoint,\lambda) &:= \{ \vertex \in \vertices \mid \breakpoint(\vertex) \leq \lambda \} \text{ and } \vertices_\sink(\breakpoint,\lambda) := \{ \vertex \in \vertices \mid \breakpoint(\vertex) > \lambda \}.
\end{align*}
We call a parameter value~$\lambda$ a \emph{breakpoint} if~$\breakpoint(\vertex)=\lambda$ for at least one vertex~$\vertex$.
The \emph{monotone parametric minimum cut (\MPMC) problem} can hence be stated as follows:
Given a monotone parametric flow network~$\graph$, compute a breakpoint function~$\breakpoint$ for~$\graph$.

The nesting property can be exploited when computing the breakpoints in ascending order:
Once  a sink-minimal minimum cut~$\cut_1=(\vertices^1_\source,\vertices^1_\sink)$ for some parameter value~$\lambda_1$ is known, it is sufficient to consider the source-contracted graph~$\graph[\lambda_2]/\vertices^1_\source$ for all~$\lambda_2>\lambda_1$.
This is because the sink-minimal minimum cut~$\cut_2$ in~$\graph[\lambda_2]$ cannot separate~$\vertices^1_\source$ and is therefore preserved in~$\graph[\lambda_2]/\vertices^1_\source$.
We give a formal proof of this in Section~\ref{sec:proof:contract}.

\section{Related Work}
\subsection{Static Max-Flow}
Recent decades have seen tremendous theoretical progress on reducing the worst-case runtime for the static max-flow problem (see~\cite{CL23} for an overview).
However, our focus is on algorithms that have been evaluated in practice.
On general instances, the push-relabel~(PRF) algorithm~\cite{GT88} and Hochbaum's pseudoflow algorithm~(HPF)~\cite{Hoc08} usually exhibit the best performance~\cite{CG97,CH09}.
PRF works by maintaining a preflow and distance labeling on the residual graph.
In each step, it selects a vertex with non-zero excess and attempts to push the excess along an admissible edge.
If no such edge exists, the distance label of the vertex is increased.
The worst-case runtime depends on the vertex selection rule.
The variant that chooses a vertex with a maximal distance label achieves a runtime of~$\mathcal O(n^2 \sqrt{m})$~\cite{CM88} and is fast in practice~\cite{CG97}.
This can be improved to~$\mathcal O(nm \log(n^2/m))$ with the dynamic trees data structure~\cite{ST83}, but experimental evidence~\cite{TW09} suggests that this only pays off on instances with very long residual paths.
HPF, which maintains a pseudoflow instead of a preflow, runs in~$\mathcal O(n^2m)$ time without dynamic trees and in~$\mathcal O(nm \log n)$ with them.

Another line of research focuses on exploiting structural features of instances that arise in computer vision applications.
Boykov and Kolmogorov~\cite{BK04} observe that these instances are often grid-like and have very short residual~$\source$-$\sink$-paths.
Their BK algorithm exploits this by growing trees from~$\source$ and~$\sink$ simultaneously.
Although no strongly polynomial time bound is known, BK is often much faster than PRF and competitive with HPF on computer vision instances~\cite{FHM16}.
Incremental breadth-first search (IBFS)~\cite{GHKTW11} modifies BK by maintaining shortest-path trees, which allows it to achieve a runtime bound of~$\mathcal O(n^2m)$ while also performing well in practice.
A version with dynamic trees runs in~$\mathcal O(nm \log n)$ time.
Excesses IBFS (EIBFS)~\cite{GHKKTW15} is a variant of IBFS that maintains preflows, which makes it more suitable for dynamic settings in which the flow network changes in discrete steps.
However, the steps must be known in advance, which is not the case in the \MPMC problem, so EIBFS is not directly applicable there.

\subsection{Monotone Parametric Max-Flow}
Hochbaum~\cite{Hoc08} distinguishes between two variants of the \MPMC problem:
In the \emph{simple sensitivity analysis}, the breakpoints are given online as a sequence~$\lambda_1<\lambda_2<\dots<\lambda_k$.
The objective is to compute a maximum flow~$\flow[\lambda_i]$ and a corresponding minimum cut~$\cut[\lambda_i]$ for every parameter value~$\lambda_i$.
In the \emph{complete parametric analysis}, the breakpoints are not known in advance.
The objective is to compute the breakpoint function~$\breakpoint$.
In contrast to our definitions in Section~\ref{sec:preliminaries}, it is not required that the minimum cut~$\cut(\breakpoint,\lambda)$ for a parameter value~$\lambda$ is sink-minimal, as long as the cuts represented by~$\breakpoint$ are nested.

\subparagraph*{Simple Sensitivity Analysis.}
\label{sec:existing:simple}
The parametric variant of PRF~\cite{GGT89}, in which the breakpoints are given, processes the parameter sequence in ascending order by exploiting the fact that PRF can be initialized with any preflow.
A maximum flow for some parameter value~$\lambda_i$ can be turned into a preflow for~$\lambda_{i+1}$ by decreasing the flow on all sink-incident edges to match their capacity and increasing the flow on all cut edges to saturate them.
This is a preflow because it produces positive excess at the tail and head vertices, respectively.
PRF is then initialized with this preflow and run for~$\lambda_{i+1}$.
For~$\mathcal O(n)$ breakpoints, this has the same asymptotic runtime as the basic PRF algorithm.
A similar approach works for the pseudoflow algorithm~\cite{Hoc08}.

\subparagraph*{Dichotomic Scheme.}
For the complete parametric analysis, Eisner and Severance~\cite{ES76} propose a \emph{dichotomic scheme} (\dichotomicscheme), which recursively subdivides the interval~$[\minLambda, \maxLambda]$.
It begins by computing minimum cuts~$\cut_\ell=(\vertices^\ell_\source,\vertices^\ell_\sink)$ for~$\minLambda$ and $\cut_u=(\vertices^u_\source,\vertices^u_\sink)$ for~$\maxLambda$.
Then it contracts~$\vertices^\ell_\source$ into~$\source$ and~$\vertices^u_\sink$ into~$\sink$, which yields a new graph~$\contractedGraph$, and calls the recursive procedure~$\Bisect([\minLambda,\maxLambda], \cut_\ell, \cut_u, \contractedGraph)$.
This procedure takes as input an interval~$[\lambda_\ell, \lambda_u]$, minimum cuts~$\cut_\ell=(\vertices^\ell_\source,\vertices^\ell_\sink)$ for~$\lambda_\ell$ and~$\cut_u=(\vertices^u_\source,\vertices^u_\sink)$ for~$\lambda_u$, and a contracted graph~$\contractedGraph$.
First, $\Bisect$ identifies the parameter value~$\lambda_m \in [\lambda_\ell,\lambda_u]$ for which~$\cut_\ell$ and~$\cut_u$ have equal capacity.
Then it invokes a max-flow algorithm on~$\contractedGraph[\lambda_m]$ to compute a minimum cut~$\cut_m=(\vertices^m_\source,\vertices^m_\sink)$.
If~$\cut_m$ has a smaller capacity than~$\cut_\ell$ when evaluated for~$\lambda_m$, then~$\cut_m$ is added to the solution set.
In this case, the algorithm makes the recursive calls~$\Bisect([\lambda_\ell,\lambda_m], \cut_\ell, \cut_m, \contractedGraph_\ell)$ and~$\Bisect([\lambda_m,\lambda_u], \cut_m, \cut_u, \contractedGraph_u)$, where~$\contractedGraph_\ell$ is obtained from~$\contractedGraph$ by contracting~$\vertices^m_\sink$ into~$\sink$, and~$\contractedGraph_u$ is obtained by contracting~$\vertices^m_\source$ into~$\source$.
Overall, \dichotomicscheme makes~$\mathcal O(n)$ max-flow invocations.
If the contracted graphs~$\contractedGraph_\ell$ and~$\contractedGraph_u$ have roughly equal size in each recursive call, then the runtime is proportional to a single max-flow computation.
However, this does not hold for all instances.

\subparagraph*{GGT.}
The GGT algorithm~\cite{GGT89} modifies \dichotomicscheme to incorporate flow updates.
It uses PRF as the max-flow algorithm, although the same scheme can also be applied to HPF~\cite{Hoc08}.
The procedure~$\Bisect$ is additionally supplied with max-flows~$\flow_\ell$ for~$\lambda_\ell$ and~$\flow_u$ for~$\lambda_u$.
Two runs of PRF are performed in parallel, one on~$\contractedGraph[\lambda_m]$ and one on the reverse graph.
The forward run is initialized with the preflow obtained from~$\flow_\ell$ and the backward run with the preflow obtained from~$\flow_u$.
W.l.o.g., assume the forward run terminates first with a max-flow~$\flow_m$.
Then the backward run is halted and the cut~$\cut_m$ is computed from~$\flow_m$.
If~$|\vertices^m_\source|>|\vertices^m_\sink|$, the backward run is finished and~$\flow_m$ is replaced with the computed flow.
Finally, $\flow_m$ is given as input to the recursive~$\Bisect$ calls, along with~$\flow_\ell$ or~$\flow_u$, respectively.

Flow updates allow GGT to achieve the same asymptotic runtime as a single PRF run.
The idea is that for the recursive call whose contracted graph is larger, the flow computations can be seen as a continuation of the previous flow computations, just like in the algorithm for the simple sensitivity analysis.
Therefore, the overall runtime only depends on the time for the recursive call with the smaller graph.


\subparagraph*{Drawbacks of Existing Approaches.}
Both bisection-based approaches do not explore the parameter space in order, so they need to maintain multiple flows and contracted graphs in memory.
Additionally, the time overhead for contracting the graphs is significant in practice.
In \dichotomicscheme, every flow computation starts from scratch on the contracted graph, which is wasteful.
On the other hand, the two concurrent flow computations performed by GGT are also wasteful in the sense that only one of them ends up being used, although it is not known in advance which one.
In an experimental comparison by Babenko et al.~\cite{BG06,BDGTZ07}, \dichotomicscheme is faster than GGT on nearly all tested instances.
The only exception is a synthetic family for which~$\Bisect$ produces contracted graphs with extremely imbalanced sizes.
By contrast, the splits tend to be fairly balanced on realistic instances.
In this case, the time savings from incorporating flow updates do not outweigh the overhead for the second flow computation.

\subsection{Applications}
\label{sec:applications}
Gallo, Grigoriadis and Tarjan~\cite{GGT89} list several applications for \MPMC, such as flow sharing and zero-one fractional programming.
More recent applications can be found in computer vision and cartography.

\subparagraph*{Computer Vision.}
Kolmogorov, Boykov and Rother \cite{KBR07} discuss applications of \MPMC in computer vision. Many problems (e.g., binary image segmentation, multiview reconstruction, surface fitting) can be formulated as maximum a-posteriori estimation of binary variables in a Markov random field (MAP-MRF).
The objective is to minimize the weighted sum of a data term and a regularization term, which enforces spatial coherence.
For example, in image segmentation, the data term encodes the likelihood that an individual pixel belongs to a certain class, whereas the regularization term ensures that similar neighboring pixels are more likely to be grouped together.
If the regularization term is submodular, then minimizing the sum for every possible weight is equivalent to solving the \MPMC problem.

\subparagraph*{Polygon Aggregation.}
\begin{figure}
    \centering
    \begin{minipage}{0.5\textwidth}
\begin{tikzpicture}[line cap=round,line join=round,>=triangle 45, xscale=1]

\coordinate (A) at (1.5*2,3);
\coordinate (B) at (1.5*2,1.5);
\coordinate (C) at (1.5*3,2.5);
\coordinate (D) at (1.5*3,4);
\coordinate (E) at (1.5*4.5,3);
\coordinate (F) at (1.5*5,1.5);
\coordinate (G) at (1.5*6,2);
\coordinate (H) at (1.5*6,3.5);

\fill[faceColor] (A) to (B) to (C) to (D) to cycle;
\fill[faceColor] (E) to (F) to (G) to (H) to cycle;

\draw [boundaryLine] (A) to (B);
\draw [boundaryLine] (B) to (C);
\draw [nonBoundaryLine] (C) to (D);
\draw [boundaryLine] (D) to (A);
\draw [nonBoundaryLine] (E) to (F);
\draw [boundaryLine] (F) to (G);
\draw [boundaryLine] (G) to (H);
\draw [boundaryLine] (H) to (E);
\draw [nonBoundaryLine] (D) to (H);
\draw [boundaryLine] (D) to (E);
\draw [nonBoundaryLine] (E) to (C);
\draw [boundaryLine] (C) to (F);
\draw [nonBoundaryLine] (F) to (B);

\node at (1.5*5.4,2.6) {$p_1$};
\node at (1.5*4.6,3.4) {$p_2$};
\node at (1.5*3.4,3.2) {$p_3$};
\node at (1.5*4.2,2.4) {$p_4$};
\node at (1.5*3.2,1.8) {$p_5$};
\node at (1.5*2.6,2.8) {$p_6$};

\end{tikzpicture}
\end{minipage}\begin{minipage}{0.5\textwidth}
\centering
\begin{tikzpicture}[line cap=round, line join=round, >=triangle 45,xscale=1.3]

\coordinate (p1) at (1.5*5.4,2.6);
\coordinate (p2) at (1.5*4.6,3.4);
\coordinate (p3) at (1.5*3.4,3.2);
\coordinate (p4) at (1.5*4.2,2.4);
\coordinate (p5) at (1.5*3.2,1.8);
\coordinate (p6) at (1.5*2.6,2.8);
\coordinate (s) at (6,0.7);
\coordinate (t) at (6,4.6);


\draw[cutEdge, terminalEdge, bend left=30] (t) to (p1);
\draw[regularEdge, terminalEdge, bend left=15] (t) to (p2);
\draw[cutEdge, terminalEdge, bend right=10] (t) to (p3);
\draw[cutEdge, terminalEdge, bend left=10] (t) to (p4);
\draw[regularEdge, terminalEdge, bend right=45] (t) to (p5);
\draw[cutEdge, terminalEdge, bend right=30] (t) to (p6);
\draw[regularEdge, terminalEdge,bend right=45] (s) to (p1);
\draw[cutEdge, terminalEdge, bend right=15] (s) to (p2);
\draw[regularEdge, terminalEdge,bend left=10] (s) to (p3);
\draw[regularEdge, terminalEdge,bend right=10] (s) to (p4);
\draw[cutEdge, terminalEdge, bend left=15] (s) to (p5);
\draw[regularEdge, terminalEdge,bend left=45] (s) to (p6);
\draw[cutEdge] (p6) to (p5);
\draw[cutEdge] (p1) to (p2);
\draw[cutEdge] (p2) to (p3);
\draw[regularEdge] (p3) to (p4);
\draw[cutEdge] (p4) to (p5);
\draw[regularEdge] (p1) to (p4);
\draw[regularEdge] (p3) to (p6);

\node [sinkNode] at (p1) {};
\node at (p1) [right=2pt] {$v_1$};
\node [sourceNode] at (p2) {};
\node at (p2) [below left] {$v_2$};
\node [sinkNode] at (p3) {};
\node [below left=1pt and -1pt of p3] {$v_3$};
\node [sinkNode] at (p4) {};
\node at (p4) [above right] {$v_4$};
\node [sourceNode] at (p5) {};
\node at (p5)  [below left] {$v_5$};
\node [sinkNode] at (p6) {};
\node at (p6) [left=2pt] {$v_6$};
\node [sinkNode] at (s) {};
\node at (s) [below=2pt] {$t$};
\node [sourceNode] at (t) {};
\node at (t) [above=2pt] {$s$};

\end{tikzpicture}
\end{minipage}
    \caption{Left: A polygon aggregation instance with input polygons in gray, a triangulation in white and a possible solution delineated in blue.
    Right: The corresponding flow network, in which vertices represent faces and edges represent boundaries.
    Source component vertices in black; sink component vertices in blue; cut edges in orange.}
    \label{fig:aggregation}
\end{figure}
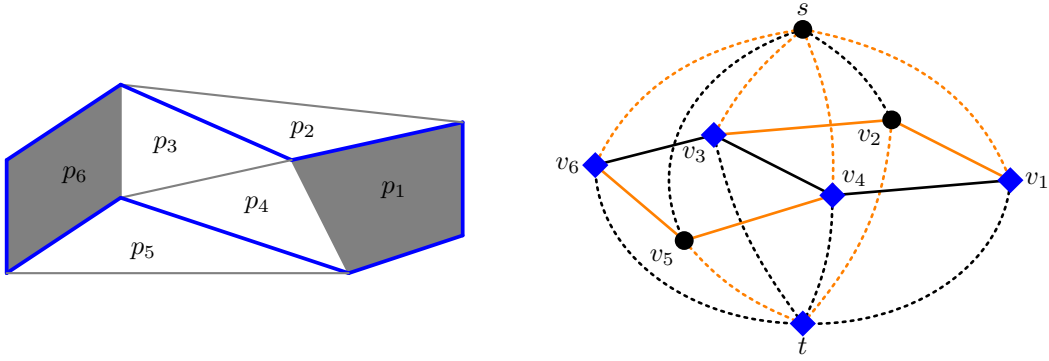

A more recently proposed application is polygon aggregation for map generalization, where the aim is to simplify a set of polygons by grouping them and replacing every group with a single representative object.
In the variant discussed by Rottmann et al.~\cite{RDGRH21}, we are given a set~$\polygons$ of input polygons in $\mathbb{R}^2$ and a triangulation~$\triangulation$ of~$\convexHull(\polygons) \setminus \polygons$, where~$\convexHull(\polygons)$ is the convex hull of~$\polygons$.
A solution $\solution=\polygons\cup\triangles$ adds a selection~$\triangles\subseteq\triangulation$ of triangles to the input~$\polygons$.
Let~$\area(\solution)$ denote the area of~$\solution$ and~$\perimeter(\solution)$ the perimeter.
The objective is to find a set~$\solutions$ of solutions such that for every~$\lambda\in [0,\infty)$, there is a solution~$\solution\in\solutions$ that minimizes~$f_\lambda(\solution)=\lambda \cdot \area(\solution) + \perimeter(\solution)$.

The input can be transformed into a monotone parametric flow network by taking the dual graph of the subdivision $X=\triangulation \cup \polygons$ and adding a source vertex~$\source$ and sink vertex~$\sink$, which are connected to all other vertices (see Figure~\ref{fig:aggregation}).
The capacity of a source-incident edge~$(\source,\vertex_i)$ is given by~$\capacity(\source,\vertex_i)=\lambda \cdot \area(p_i)+\ell(p_i)$, where $\ell (p_i)$ is the length of the shared boundary of~$p_i$ and the outer face.
For sink-incident edges~$(\vertex_i,\sink)$, the capacity is set to~$\infty$ if~$p_i\in\polygons$ and~$0$ otherwise.
For all other edges~$(\vertex_i,\vertex_j)$, the capacity is~$\capacity(\vertex_i,\vertex_j)=\ell (p_i,p_j)$, where $\ell (p_i,p_j)$ is the length of the common boundary of the faces $p_i,p_j\in X$.
Let~$\breakpoint$ be the breakpoint function for this \MPMC instance.
For every~$\lambda\in[0,\infty)$, an optimal solution is given by the input polygons and triangles that correspond to the sink component~$\vertices_\sink(\breakpoint,\lambda)$.

Rottmann et al.\ model the cost function as~$f_\alpha(\solution)=\alpha\cdot \area(\solution) + (1-\alpha) \cdot \perimeter(\solution)$ with~$\alpha\in[0,1]$.
Since~$f_\alpha(\solution)=(1-\alpha)\cdot f_\lambda(\solution)$ holds for~$\lambda=\frac{\alpha}{(1-\alpha)}$, this is equivalent to our formulation: the solution sets are the same, although the parameter values for which they are optimal are shifted.
However, when using~$f_\alpha$, the parametric capacities are not monotone.
Consequently, Rottmann et al.\ do not notice the equivalence of the problem to \MPMC.
Instead, they prove the nesting property directly.
To solve it, they use \dichotomicscheme without contraction and with PRF as the max-flow algorithm.
Because the number of breakpoints is very high, this algorithm does not scale well.
To achieve reasonable runtimes on larger instances, the authors compute an~$\varepsilon$-approximation of the solution set: for each~$\alpha\in[0,1]$, the relative difference between the capacity of the minimum cut and the capacity of the computed cut must be at most~$\varepsilon$.
\dichotomicscheme can be adapted for this scenario:
When~$\Bisect$ computes a minimum cut~$\cut_m$ for the intersection point~$\lambda_m$, it is compared to the cut~$\cut_\ell$ associated with~$\lambda_\ell$.
If the capacity of~$\cut_\ell$ for~$\lambda_m$ is at most~$(1 + \varepsilon)$ times larger than that of~$\cut_m$, the recursive~$\Bisect$ calls are skipped.

\section{Parametric BFS}
\label{sec:alg}
We present \emph{parametric breadth-first search} (\PBFS), a new algorithm for complete parametric analysis that addresses the shortcomings of GGT.
In particular, it discovers the breakpoints in ascending order, which allows it to incorporate flow updates in a straightforward manner.
As described in Section~\ref{sec:existing:simple}, a maximum flow~$\flow[\lambda_i]$ for a breakpoint~$\lambda_i$ can be turned into a preflow for any parameter value~$\lambda>\lambda_i$ by saturating the cut edges and ensuring that the sink-incident edges do not exceed their capacity.
The next breakpoint~$\lambda_{i+1}$ is the first value for which it is no longer possible to redistribute the resulting excesses without disconnecting the sink component of the residual graph.
PRF cannot easily identify this value because it only examines the neighborhood of a vertex to decide where its excess is pushed.

For this reason, \PBFS switches to an approach inspired by IBFS, which grows shortest-path trees from~$\source$ and~$\sink$ in the residual graph.
Similarly, \PBFS maintains a reverse shortest-path tree~$\tree$ rooted in~$\sink$ that only uses residual edges.
This ensures that it is always possible to send a non-zero amount of excess towards~$\sink$ by following tree edges.
In step~$i$, the algorithm identifies the smallest parameter value~$\lambda_{i+1}$ for which this strategy saturates at least one tree edge.
At this point, the residual graph changes and~$\tree$ is disconnected into a forest.
It is repaired using the \emph{vertex adoption} method from IBFS.
If a vertex~$\vertex$ cannot be reattached because there is no residual path to~$\sink$ left, then~$\lambda_{i+1}$ is the next breakpoint and~$\vertex$ enters the source component.
The algorithm exploits the nesting property by contracting~$\vertex$ into~$\source$, which removes the need for maintaining a shortest-path tree in the source component as well.

\subparagraph*{Algorithm Description.}
\begin{algorithm}
    \caption{Parametric BFS.}\label{alg:main}
    \Input{Monotone parametric flow network $\graph = (\vertices,\edges,\capacity,[\minLambda,\maxLambda],\source,\sink)$}
    \Output{Breakpoint function $\breakpoint\colon \vertices \to \mathbb{R}_{\geq 0}$}
    \BlankLine
    $\initialFlow \gets \Initialize()$\;
    $\CalcFlowFunction(\initialFlow)$\;
    $i \gets 0, \lambda_0 \gets \minLambda$\;
    \While{$\lambda_i < \maxLambda$}{
        $\lambda_{i+1} \gets \min_{\edge \in \treeEdges} \rootLambda(\edge)$\label{alg:calculateNextParameter}\;
        \lIf{$\lambda_{i+1} > \maxLambda$}{\Return $\breakpoint$\label{alg:return}}
        $\UpdateResidualGraph()$\;
        $\ReconnectTree()$\;
        $\DrainExcess(\lambda_{i+1})$\;
        $i \gets i+1$\;
    }
    \Return $\breakpoint$\;
\end{algorithm}
Pseudocode for \PBFS is given in Algorithm~\ref{alg:main}, with subroutines in Algorithms~\ref{alg:init}--\ref{alg:adopt}.
Only some data structures are maintained explicitly; the others are implicit.
In step~$i$ with current parameter value~$\lambda_i$, the algorithm maintains a sink component~$\vertices_\sink$ such that~$\cut=(\vertices_\source,\vertices_\sink)$ with~$\vertices_\source:=\vertices\setminus\vertices_\sink$ is the sink-minimal minimum cut in~$\graph[\lambda_i]$.
The source-contracted graph is given by~$\contractedGraph:=\graph[\lambda_i]/\vertices_\source=(\contractedVertices,\contractedEdges,\capacity[\lambda_i],\source,\sink)$ with~$\contractedVertices:=\vertices_\sink \cup \{\source\}$ and~$\contractedEdges:=\edges/\vertices_\source$.
The algorithm maintains a flow function~$\flow(\edge)$ on each edge~$\edge\in\contractedEdges$ such that~$\flow[\lambda_i]$ is a maximum flow on~$\contractedGraph$.
The residual graph~$\residualContractedGraph:=(\contractedVertices, \residualContractedEdges, \residualCapacity[\lambda_i],\source,\sink)$ with residual capacity function~$\residualCapacity(\edge):=\capacity(\edge) - \flow(\edge)$ is represented implicitly by the set~$\residualContractedEdges$ of residual edges.
Finally, the algorithm maintains an (unweighted) shortest-path tree~$\tree=(\treeVertices,\treeEdges)$ on~$\residualContractedGraph$ and a distance labeling~$\distanceLabel$.

\begin{algorithm}
    \caption{Initialization of \PBFS.}
    \label{alg:init}
    \myproc{$\Initialize()$}{
        $\initialFlow \gets \CalcMaxFlow(\graph[\minLambda])$\;
        $\residualGraph \gets \BuildResidualGraph(\graph,\initialFlow)$\;
        $\tree=(\treeVertices,\treeEdges) \gets \ReverseBFS(\residualGraph,\sink)$\;
        $\residualContractedEdges \gets \edges(\residualGraph)/(\vertices \setminus \treeVertices)$\;
        \ForEach{$\vertex \in \vertices \setminus \treeVertices$}{
            $\breakpoint(\vertex) \gets \minLambda$\;
            $\distanceLabel(\vertex) \gets \infty$\;
        }
        \ForEach{$\vertexB \in \treeVertices$}{
            $\breakpoint(\vertexB) \gets \infty$\;
            $\distanceLabel(\vertexB) \gets \dist_\tree(\vertexB,\sink)$\;
            $\outEdges(\vertexB) \gets \{(\vertexB,\vertexA) \in \edges\}\label{alg:init:outEdges}$\;
        }
        \Return $\initialFlow$\;
    }
\end{algorithm}

\begin{algorithm}
    \caption{Flow function initialization.}\label{alg:calcflow}
    \myproc{$\CalcFlowFunction(\initialFlow)$\label{alg:calcflow:begin}}{
        $\flow \gets \initialFlow$\;
        \lForEach{$\vertex\in\vertices$}{$\excess(\vertex) \gets 0$}
        \lForEach{$\edge \in \edges$}{$\rootLambda(\edge) \gets \infty$}
        $\excessVertices \gets \emptyset$\;
        \ForEach{$\edge = (\source,\vertexB) \in \edges$ with~$\initialFlow(\edge)=\capacity(\edge)(\minLambda)$}{
            $\excess(\vertexB) \gets \excess(\vertexB) + \capacity(\edge) - \flow(\edge)$\;
            $\excessVertices \gets \excessVertices \cup \{\vertexB\}$\;
            $\SetFlow(\flow, \edge, \capacity(\edge))$\;
        }
        \ForEach{$\edge = (\vertexB,\sink) \in \edges$ with~$\initialFlow(\edge)=\capacity(\edge)(\minLambda)$}{
            $\excess(\vertexB) \gets \excess(\vertexB) + \flow(\edge) - \capacity(\edge)$\;
            $\excessVertices \gets \excessVertices \cup \{\vertexB\}$\;
            $\SetFlow(\flow, \edge, \capacity(\edge))$\;
        }
        $\DrainExcess(\minLambda)$\;
        \Return $\flow$\label{alg:calcflow:end}\;
    }
    \myproc{$\SetFlow(\flow,(\vertexA,\vertexB),\flow')$}{
        $\flow(\vertexA,\vertexB) \gets \flow'$\;
        $\flow(\vertexB,\vertexA) \gets -\flow'$\;
    }
\end{algorithm}

The algorithm begins with the procedure~$\Initialize$ (pseudocode in Algorithm~\ref{alg:init}), which computes a max-flow~$\initialFlow$ for~$\minLambda$ and the corresponding set~$\residualContractedEdges$ of residual edges.
The shortest-path tree~$\tree$, distance labeling~$\distanceLabel$ and sink component~$\vertices_\sink$ are computed with a reverse BFS from~$\sink$ in the residual graph.
Alternatively, if IBFS is used as the max-flow algorithm, it already computes them alongside~$\initialFlow$.
The breakpoint of each vertex in the source component is set to~$\minLambda$.
The flow functions~$\flow$ are then calculated in~$\CalcFlowFunction$ (Algorithm~\ref{alg:calcflow}).
Starting with the constant flow~$\initialFlow$, the flow functions of all saturated source- and sink-incident edges are set to their capacity functions, which ensures that they remain saturated.
For parameter values above~$\minLambda$, this creates an excess at each vertex~$\vertex$, which is stored in the excess function~$\excess(\vertex)$.
The excesses are drained by sending them towards~$\sink$ in~$\tree$.
This is done by maintaining a set~$\excessVertices$ of vertices whose excess function was modified.

The procedure~$\DrainExcess$ (lines~\ref{alg:drain:begin}--\ref{alg:drain:end} of Algorithm~\ref{alg:reconnect-drain}) processes these vertices in decreasing order of distance to~$\sink$.
For each vertex~$\vertexA\in\excessVertices$ with corresponding tree edge~$\edge=(\vertexA,\vertexB)$, the algorithm pushes the excess across~$\edge$ to~$\vertexB$.
After all vertices have been processed, all excess functions except for that of~$\sink$ are~$0$ and the flow functions on tree edges are non-decreasing.
Because all tree edges are residual, there are parameter values~$\lambda>\minLambda$ for which~$\flow[\lambda]$ is still a flow.
For each edge~$\edge\in\contractedEdges$, the algorithm maintains the~\emph{flow limit}~$\rootLambda(\edge)$, which is the smallest parameter value for which~$\edge$ changes from residual to saturated.
The algorithm upholds the invariant that this is~$\infty$ for all non-tree edges.
When~$\DrainExcess(\lambda)$ changes the flow on a tree edge~$\edge$, then~$\rootLambda(\edge)$ is recalculated by calling the procedure~$\CalcRoot(\capacity(\edge)-\flow(\edge),\lambda)$, which returns the smallest root of the residual capacity function~$\capacity(\edge)-\flow(\edge)$ in the interval~$[\lambda,\infty)$, or~$\infty$ if there is no root.

\begin{algorithm}
    \caption{Residual graph update.}\label{alg:update}
    \myproc{$\UpdateResidualGraph()$\label{alg:update:begin}}{
        $\revEdgesSat \gets \{ \edge=(\vertexA,\vertexB) \mid (\vertexB,\vertexA)\in\treeEdges\land \flow(\edge)(\lambda_i)=\capacity(\edge)(\lambda_i)\}$\;
        $\revEdges \gets \{ \edge \in \revEdgesSat \mid \flow(\edge)(\lambda_{i+1}) < \capacity(\edge)(\lambda_{i+1}) \}$\;
        $\saturatedEdges \gets \{ \edge \in \treeEdges \mid \rootLambda(\edge)=\lambda_{i+1}\}$\;
        $\residualContractedEdges \gets \residualContractedEdges \cup \revEdges \setminus \saturatedEdges$\;
        $\orphans,\excessVertices \gets \emptyset$\;
        \ForEach{$\edge = (\vertexA,\vertexB) \in \saturatedEdges$}{
            $\RemoveTreeEdge(\edge)$\label{alg:update:end}\;
        }
    }
    \myproc{$\RemoveTreeEdge(\edge=(\vertexA,\vertexB))$}{
        $\newFlow \gets \capacity(\edge) + \flow(\edge)(\lambda_{i+1}) - \capacity(\edge)(\lambda_{i+1})$\;
        $\excess(\vertexA) \gets \excess(\vertexA) + \flow(\edge) - \newFlow$\;
        $\excess(\vertexB) \gets \excess(\vertexB) + \newFlow - \flow(\edge)$\;
        $\SetFlow(\flow, \edge, \newFlow)$\;
        $\rootLambda(\edge) \gets \infty$\;
        $\orphans \gets \orphans \cup \{ \vertexA \}$\;
        $\excessVertices \gets \excessVertices \cup \{ \vertexA, \vertexB \}$\;
        $\treeEdges \gets \treeEdges \setminus \{ \edge \}$\;
    }
\end{algorithm}
    
As long as the current parameter value~$\lambda_i$ is smaller than~$\maxLambda$, the algorithm performs another step.
Line~\ref{alg:calculateNextParameter} of Algorithm~\ref{alg:main} calculates the next parameter value~$\lambda_{i+1}$ as the minimal flow limit among all tree edges.
$\UpdateResidualGraph$ (Algorithm~\ref{alg:update}) identifies the set~$\revEdges$ of edges that were saturated for~$\lambda_i$ but are residual for~$\lambda_{i+1}$, which must be reverse edges of tree edges, and adds them to~$\residualContractedEdges$.
Additionally, it removes all saturated tree edges from~$\tree$ and~$\residualContractedEdges$.
Because flow can no longer be sent across a removed edge~$\edge=(\vertexA,\vertexB)$, excess starts building up at~$\vertexA$ and~$\vertexB$, which is stored in the excess functions~$\excess(\vertexA)$ and~$\excess(\vertexB)$.
The vertex~$\vertexA$, which no longer has a residual path to~$\sink$, is called an~\emph{orphan} and stored in a set~$\orphans$.
At this point, $\flow$ is still a flow for~$\lambda_{i+1}$ but only a pseudoflow beyond that, and~$\tree$ is a forest with more than one connected component.

\begin{algorithm}
    \caption{Tree repairing, excess draining.}\label{alg:reconnect-drain}
    \myproc{$\ReconnectTree()$\label{alg:reconnect:begin}}{
        \While{$\orphans \neq \emptyset$}{
            $\vertexA \gets \Pop(\orphans)$\;
            \lIf{$\AdoptWithSameDist(\vertexA)$}{\Continue}
            \For {$\vertexB \in \children_\tree(\vertexA)$}{
                $\RemoveTreeEdge((\vertexB,\vertexA))$\;
            }
            \lIf{$\AdoptWithNewDist(\vertexA)$}{\Continue}
            $\breakpoint(\vertexA) \gets \lambda_{i+1}$\;
            $\treeVertices \gets \treeVertices \setminus \{ \vertexA \}$\;
            $\residualContractedEdges \gets \residualContractedEdges/\{\vertexA\}$\label{alg:reconnect:end}\;
        }
    }
    \myproc{$\DrainExcess(\lambda)$\label{alg:drain:begin}}{
        \ForEach{$\edge=(\vertexA,\vertexB)\in\treeEdges$ with~$\vertexA\in\excessVertices$ in descending order of $\distanceLabel(\vertexA)$}{
            $\SetFlow(\flow, \edge, \flow(\edge) + \excess(\vertexA))$\;
            $\rootLambda(\edge) \gets \CalcRoot(\capacity(\edge)-\flow(\edge),\lambda)$\;
            $\excess(\vertexB) \gets \excess(\vertexB) + \excess(\vertexA)$\;
            $\excess(\vertexA) \gets 0$\;
            $\excessVertices \gets \excessVertices \cup \{ \vertexB \}$\label{alg:drain:end}\;
        }
    }
\end{algorithm}

\begin{algorithm}
    \caption{Vertex adoption.}
    \label{alg:adopt}
    \myproc{$\AdoptWithSameDist(\vertexA)$}{
        \While{$\outEdges(\vertexA)\neq\emptyset$}{
            $\edge = (\vertexA,\vertexB) \gets \Pop(\outEdges(\vertexA))$\;
            \If{$(\vertexA,\vertexB)\in\residualContractedEdges$ and $\distanceLabel(\vertexA)=\distanceLabel(\vertexB)+1$}{
                $\treeEdges \gets \treeEdges \cup \{\edge\}$\;
                \Return \textbf{true}\;
            }
        }
        \Return \textbf{false}\;
    }
    \myproc{$\AdoptWithNewDist(\vertexA)$}{
        \lIf{$\{(\vertexA,\vertexB) \in \residualContractedEdges\} = \emptyset$}{\Return \textbf{false}}
        $\edge = (\vertexA,\vertexC)\gets \argmin_{(\vertexA,\vertexB) \in \residualContractedEdges} \distanceLabel(\vertexB)$\;
        \lIf{$\distanceLabel(\vertexC)\geq |\treeVertices|-1$}{\Return \textbf{false}}
        $\distanceLabel(\vertexA) \gets \distanceLabel(\vertexC)+1$\;\label{alg:increaseDistanceLabel}
        $\treeEdges \gets \treeEdges \cup \{\edge\}$\;
        $\outEdges(\vertexA) \gets \{(\vertexA,\vertexB) \in \edges\}$\label{alg:adopt:outEdges}\;
        \Return \textbf{true}\;
    }
\end{algorithm}

$\ReconnectTree$ (lines~\ref{alg:reconnect:begin}--\ref{alg:reconnect:end} of Algorithm~\ref{alg:reconnect-drain}) rebuilds the tree structure so that the excesses can once again be drained.
This step is done exactly as in IBFS.
For each orphan~$\vertexA$, the algorithm tries to reconnect~$\vertexA$ to the tree by searching for an outgoing residual edge~$\edge=(\vertexA,\vertexB)$ such that the distance label~$\distanceLabel(\vertexB)$ is minimal.
This search is split into two procedures (pseudocode in Algorithm~\ref{alg:adopt}):
First, $\AdoptWithSameDist(\vertexA)$ checks whether there is an outgoing residual edge that is already admissible.
In this case, all vertices in the subtree rooted in~$\vertexA$ can keep their current distance label.
Otherwise, the distance label of~$\vertexA$ must increase and the subtree rooted in~$\vertexA$ may need to be split apart.
Therefore, all incoming tree edges of~$\vertexA$ are removed and their tail vertices are made orphans.
Then, $\AdoptWithNewDist(\vertexA)$ searches for an outgoing residual edge~$\edge=(\vertexA,\vertexC)$ such that the distance label~$\distanceLabel(\vertexC)$ is minimal.
If no such edge exists, there is no longer a residual~$\vertexA$-$\sink$-path, so~$\vertexA$ is moved into the source component and its breakpoint is set to~$\lambda_{i+1}$.
Otherwise, $\edge$ is added to the tree and the distance label of~$\vertex$ is set to~$\distanceLabel(\vertexC)+1$, which ensures that~$\edge$ is admissible.
Once all orphans have been processed, $\tree$ is once again a tree, so~$\DrainExcess(\lambda_{i+1})$ is called to drain the excesses towards~$\sink$.

\subparagraph*{Performance Optimizations.}
To achieve the desired runtime bound, the algorithm must ensure that~$\AdoptWithSameDist(\vertexA)$ examines each outgoing edge of a vertex~$\vertexA$ at most once before~$\distanceLabel(\vertexA)$ increases.
In the pseudocode, this is done by maintaining a set~$\outEdges(\vertexA)$ of unprocessed edges, which is initialized with all outgoing edges in line~\ref{alg:init:outEdges} of Algorithm~\ref{alg:init}.
During~$\AdoptWithSameDist(\vertexA)$, examined edges are removed from~$\outEdges(\vertexA)$ to ensure that they are not looked at again.
When~$\distanceLabel(\vertexA)$ is increased in line~\ref{alg:adopt:outEdges} of Algorithm~\ref{alg:adopt}, $\outEdges(\vertexA)$ is reset.
Our implementation of \PBFS does not maintain~$\outEdges(\vertexA)$ explicitly.
Rather, it follows the approach of IBFS by maintaining a linked list of all outgoing edges of~$\vertexA$ in~$\graph$, as well as a \emph{current edge} pointer to the first unprocessed edge in the list.
Whenever~$\outEdges(\vertexA)$ is reset, this is achieved by setting the current edge pointer to the first outgoing edge of~$\vertexA$.
When an edge is processed, the pointer is advanced to the next edge.

As noted by Goldberg et al.~\cite{GHKKTW15}, the depicted adoption strategy sometimes performs badly in practice.
When an orphan~$\vertexA$ is adopted by a vertex that later becomes an orphan itself, then~$\vertexA$ is processed again.
In some cases, this may cause the same vertex to be processed many times.
To prevent this, Goldberg et al.\ propose a three-pass adoption strategy, which ensures that each orphan is processed at most three times during each~$\ReconnectTree$ call, as well as a hybrid strategy that starts with the original adoption scheme but switches to the three-pass scheme once the amount of performed work per orphan exceeds a certain threshold.
Our implementations of IBFS and \PBFS use the hybrid strategy, which exhibits the best performance in practice.

\subparagraph*{Implementation Details.}
Line~\ref{alg:calculateNextParameter} of Algorithm~\ref{alg:main} calculates the next parameter value~$\lambda_{i+1}$ by examining the flow limits of all tree edges.
On realistic instances, this is wasteful because most iterations only modify a small portion of the tree and leave the flow limits of the rest unchanged.
Our implementation therefore maintains a priority queue on~$\treeEdges$ with the flow limit as the key.
This allows~$\lambda_{i+1}$ and the set~$\saturatedEdges$ of saturated tree edges to be computed quickly by extracting the element with minimal key until the key increases.
In the asymptotic runtime, this incurs an additional logarithmic factor because the queue must be updated every time a flow limit changes.
However, it is faster in practice unless the average number of updates per iteration is very high.

In the pseudocode, the edge set~$\residualContractedEdges$ of the contracted residual graph is maintained explicitly and updated when flows are changed or vertices are contracted.
This is merely done for ease of exposition.
In our implementation, this set is maintained implicitly.
Whenever the outgoing edges of a vertex are inspected, edges that are saturated or lead into the source component are ignored.
In particular, this means that the only task performed by~$\UpdateResidualGraph$ is identifying and removing the saturated tree edges.

\section{Proof of Correctness and Runtime}
\label{sec:proof}
We give a detailed proof of correctness and a runtime analysis for \PBFS.
We start with an overview of the main ideas in Section~\ref{sec:proof:overview}.
In Section~\ref{sec:proof:contract}, we show that the source component can be contracted whenever a new breakpoint is found.
Section~\ref{sec:proof:residual} proves invariants related to the residual graph, the shortest-path-tree and the distance labels, including the correctness of the adoption procedures.
In Section~\ref{sec:proof:functions}, we characterize the flow and excess functions maintained by the algorithm, depending on the type of edge, and show that the flow limits are updated correctly.
This is used to characterize how the residual graph changes as~$\lambda$ increases in Section~\ref{sec:proof:extend}.
These insights are then combined to give an inductive proof of correctness in Section~\ref{sec:proof:correctness} and bound the runtime in Section~\ref{sec:proof:runtime}.

\subsection{Overview}
\label{sec:proof:overview}
The algorithm maintains the invariants that~$\flow[\lambda_i]$ is a maximum flow (Lemma~\ref{lem:extendFlow}) and~$\tree$ is a shortest-path tree in the residual graph~$\residualContractedGraph$ (Corollary~\ref{cor:tree}).
Because the current cut is kept saturated, $\flow[\lambda]$ is a maximum flow for~$\lambda\in[\lambda_i,\lambda_{i+1}]$.
The algorithm ensures that only tree edges can become saturated for~$\lambda_{i+1}$ by maintaining the invariant that the residual capacity functions are non-increasing on tree edges (Lemma~\ref{lem:treeFlow}) and constant for out-of-tree edges whose endpoints are not~$\source$ or~$\sink$ (Lemma~\ref{lem:outOfTreeFlow}).
When the saturated tree edges are removed, $\tree$ becomes a forest.
New residual edges~$(\vertexB,\vertexA)$ may be introduced, but these are reverse edges of removed tree edges~$(\vertexA,\vertexB)$, which are residual.
Hence, the distance labeling remains valid (Lemma~\ref{lem:distanceLabels:valid}) because~$\distanceLabel(\vertexB)=\distanceLabel(\vertexA)-1<\distanceLabel(\vertexA)+1$ holds.
$\ReconnectTree$ ensures that~$\tree$ is once again a shortest-path tree; its correctness is proven in the same manner as for IBFS (Lemma~\ref{lem:tree} and Corollary~\ref{cor:tree}).
Vertices that cannot be reconnected do not have a residual path to~$\sink$, so they are no longer in the minimal sink component (Lemma~\ref{lem:adoptWithNewDist:fail}).
When excesses are pushed, $\flow[\lambda_{i+1}]$ is not changed because all excess functions evaluate to~$0$ for~$\lambda_{i+1}$ (Lemma~\ref{lem:flowChange}).

The runtime bound of~$\mathcal O(n^2m)$ relies on the observation that the distance label of a vertex never decreases and cannot exceed~$n-1$ (Lemma~\ref{lem:distanceLabels:nonDecreasing}).
The number of iterations in the main loop is bounded by the number of edge removals from~$\tree$, which is in~$\mathcal O(nm)$ (Lemma~\ref{lem:treeAdditions}):
When an edge~$\edge=(\vertexA,\vertexB)$ is removed from~$\tree$, this is either because~$\AdoptWithSameDist(\vertexA)$ has failed or because~$\edge$ is saturated.
In the former case, it can be shown that~$\distanceLabel(\vertexA)$ is increased in~$\AdoptWithNewDist(\vertexA)$ (Lemma~\ref{lem:increaseDistanceLabel}).
For the latter case, we observe that~$\edge$ is admissible and saturated when it is removed.
In order to re-enter~$\tree$, it must become residual again.
This requires that the reverse edge~$\edge'=(\vertexB,\vertexA)$ enters~$\tree$ in the meantime.
In order to make~$\edge'$ admissible, the distance label of~$\vertexB$ must increase.
Within an iteration, $\UpdateResidualGraph$ and~$\DrainExcess$ take time~$\mathcal O(n)$.
The initialization phase is dominated by the initial flow computation, which takes time~$\mathcal O(n^2m)$ with IBFS.
Finally, it can be proven in the same manner as for IBFS that the time spent in~$\ReconnectTree$ throughout the entire algorithm is in~$\mathcal O(nm)$ (Lemma~\ref{lem:reconnectTime}).

For IBFS, the runtime can be improved to~$\mathcal O(nm \log n)$ by using dynamic trees to represent the shortest-path trees.
Doing the same in \PBFS is challenging:
Dynamic trees support retrieving the edge with minimum weight in logarithmic time.
However, \PBFS must retrieve the edge whose residual capacity function has the smallest root. 
It is unclear whether this can also be done in logarithmic time.

\subsection{Source Contraction}
\label{sec:proof:contract}
For a vertex set~$\vertices'\subseteq\vertices$, the source-contracted graph~$\graph/\vertices'$ preserves exactly the cuts in~$\graph$ that do not separate~$\vertices'$:
For every cut~$\cut=(\vertices_\source,\vertices_\sink)$ in~$\graph$ with~$\vertices' \subseteq \vertices_\source$, there is a cut~$\cut'=(\vertices_\source \setminus \vertices', \vertices_\sink)$ in~$\graph/\vertices'$ with~$|\cut|=|\cut'|$.
Conversely, every cut~$\cut=(\vertices_\source,\vertices_\sink)$ in~$\graph/\vertices'$ has a corresponding cut~$\cut=(\vertices' \cup \vertices_\source,\vertices_\sink)$ in~$\graph$ with~$|\cut|=|\cut'|$.

For any flow~$\flow$ in~$\graph$, a corresponding flow~$\flow/\vertices'$ in~$\graph/\vertices'$ is obtained as follows: the flow over an edge~$\edge$ is the sum of the flows over all edges that are merged with~$\edge$ during the contraction.
It is easy to verify that~$\flow/\vertices'$ fulfills the flow conservation and capacity constraints and has the same value as~$\flow$.
The residual graph for~$\flow/\vertices'$ is~$\residualGraph/\vertices'$.

\begin{lemma} \label{lem:contractFlow}
Let~$\flow$ be a maximum flow in a flow network~$\graph$ and~$\vertices'$ a set of vertices such that there is a minimum cut~$\cut=(\vertices_\source,\vertices_\sink)$ in~$\graph$ with~$\vertices'\subseteq\vertices_\source$.
Then~$\flow/\vertices'$ is a maximum flow in~$\graph/\vertices'$.
\end{lemma}
\begin{proof}
By Observation~\ref{obs:saturateAllCuts}, $\flow$ saturates~$\cut$.
Then~$\flow/\vertices'$ saturates the corresponding minimal cut~$(\vertices_\source\setminus\vertices',\vertices_\sink)$ in~$\graph/\vertices'$ and is therefore maximal.
\end{proof}

\begin{lemma}
\label{lem:contract}
Let~$\graph=(\vertices,\edges,\capacity,[\minLambda,\maxLambda],\source,\sink)$ be a monotone parametric flow network and $\lambda_1,\lambda_2\in[\minLambda,\maxLambda]$ two parameter values with~$\lambda_1<\lambda_2$.
Let~$\cut_1=(\vertices^1_\source,\vertices^1_\sink)$ and~$\cut_2=(\vertices^2_\source,\vertices^2_\sink)$ be the sink-minimal minimum cuts in~$\graph[\lambda_1]$ and~$\graph[\lambda_2]/\vertices^1_\source$, respectively.
Then the sink-minimal minimum cut in~$\graph[\lambda_2]$ is~$(\vertices^1_\source \cup \vertices^2_\source,\vertices^2_\sink)$.
\end{lemma}
\begin{proof}
It follows from the nesting property (Observation~\ref{lem:nesting}) that~$\vertices^1_\source$ is fully contained in the source component of the sink-minimal minimum cut in~$\graph[\lambda_2]$.
Let~$\flow$ be a maximum flow in~$\graph[\lambda_2]$.
By Lemma~\ref{lem:contractFlow}, $\flow/\vertices^1_\source$ is a maximum flow in~$\graph[\lambda_2]/\vertices^1_\source$.
It follows from Corollary~\ref{cor:sinkMinimalCut} that~$\cut_2=\cut(\flow/\vertices^1_\source)$, i.e., $\vertices^2_\sink$ is the set of vertices from which~$\sink$ is reachable in~$\residualGraph[\lambda_2]/\vertices^1_\source$.
This set remains the same in~$\residualGraph[\lambda_2]$; in particular, $\sink$ is not reachable from~$\vertices^1_\source$ because~$\flow$ saturates~$\cut_1$ by Observation~\ref{obs:saturateAllCuts}.
Hence, by Corollary~\ref{cor:sinkMinimalCut}, $(\vertices^1_\source\cup\vertices^2_\source,\vertices^2_\sink)=\cut(\flow)$ is the sink-minimal minimum cut in~$\graph[\lambda_2]$.
\end{proof}

\subsection{Residual Edges and Sink Component Tree}
\label{sec:proof:residual}
\begin{lemma} \label{lem:tree}
    At any time during the execution of the algorithm, $\tree$ is a (directed) forest in~$\residualContractedGraph$ whose edges are all admissible and~$\orphans\subseteq\treeVertices$ is the set of vertices besides~$\sink$ that have no outgoing edge in~$\tree$.
\end{lemma}
\begin{proof}
    After~$\Initialize$, $\tree$ is a shortest-path tree in~$\residualContractedGraph$ and~$\orphans$ is empty, so the claim is true.
    If an edge~$\edge=(\vertexA,\vertexB)$ is removed from~$\residualContractedGraph$, it is also removed from~$\tree$ via~$\RemoveTreeEdge(\edge)$.
    This procedure adds~$\vertexA$ to~$\orphans$, so the invariant is upheld.
    The only procedures that add an edge~$\edge=(\vertexA,\vertexB)\in\residualContractedEdges$ to~$\tree$ are~$\AdoptWithSameDist(\vertexA)$ and~$\AdoptWithNewDist(\vertexB)$, which are called when~$\vertexA$ is removed from~$\orphans$.
    If both fail to add an edge, then~$\vertexA$ is removed from~$\treeVertices$.
    If either succeeds, then~$\edge$ is admissible.
    Finally, the only procedure that changes the distance label~$\distanceLabel(\vertexA)$ of a vertex~$\vertexA$ is~$\AdoptWithSameDist(\vertexA)$, which is only called when~$\vertexA$ is isolated in~$\tree$.
    Therefore, the edges in~$\tree$ remain admissible.
    It follows from the definition of admissibility that~$\tree$ is acyclic and therefore a forest.
\end{proof}

\begin{corollary} \label{cor:tree}
    Whenever~$\orphans$ is empty, $\tree$ is a shortest-path tree to~$\sink$ in~$\residualContractedGraph$.
\end{corollary}

\begin{lemma} \label{lem:distanceLabels:valid}
    At every point during the execution of the algorithm, the distance labeling~$\distanceLabel$ is valid for~$\residualContractedGraph$.
\end{lemma}
\begin{proof}
    The initial distance labeling is valid because each label~$\distanceLabel(\vertex)$ corresponds to the length of a shortest~$\vertex$-$\sink$-path in~$\residualContractedGraph$.
    Removing an edge or vertex from~$\residualContractedGraph$ preserves the validity of the distance labeling.
	Hence, the only critical parts of the algorithm are the addition of new edges in~$\UpdateResidualGraph$ and the modification of the distance labels in~$\AdoptWithNewDist$.
    $\UpdateResidualGraph$ only adds edges~$(\vertexB,\vertexA) \in \revEdges$, for which the reverse edge~$(\vertexA,\vertexB)$ is in~$\treeEdges$ and therefore admissible by Lemma~\ref{lem:tree}.
    This implies~$\distanceLabel(\vertexB)=\distanceLabel(\vertexA)-1<\distanceLabel(\vertexA)+1$.
    Consider the update of a distance label~$\distanceLabel(\vertexA)$ to~$\distanceLabel(\vertexC)+1$ in~$\AdoptWithNewDist(\vertexA)$.
	Because~$\distanceLabel(\vertexC)$ was chosen to be minimal among all neighbors of~$\vertexA$, the distance label of~$\vertexA$ remains valid.
\end{proof}

\begin{lemma} \label{lem:distanceLabels:nonDecreasing}
    The distance label~$\distanceLabel(\vertexA)$ of a vertex~$\vertexA$ never decreases and increases at most~$n$ times.
\end{lemma}
\begin{proof}
    The only procedure that changes~$\distanceLabel(\vertexA)$ is~$\AdoptWithNewDist(\vertexA)$, which sets it to~$\distanceLabel(\vertexC)+1$.
    Before this happens, $\distanceLabel(\vertexA)\leq\distanceLabel(\vertexC)+1$ must already hold because~$(\vertexA,\vertexC)\in\residualContractedEdges$ and the distance labeling is valid by Lemma~\ref{lem:distanceLabels:valid}.
    Hence, $\distanceLabel(\vertexA)$ does not decrease.
    Because~$\distanceLabel(\vertexA)$ is a lower bound for the distance from~$\vertexA$ to~$\sink$, its value can range from~$0$ to~$n-1$ and therefore increase at most~$n$ times.    
\end{proof}

\begin{lemma} \label{lem:outedges}
    At any time during the execution of the algorithm, the set~$\outEdges(\vertexA)$ for a vertex~$\vertexA \in \treeVertices$ contains all admissible outgoing edges in~$\residualContractedEdges \setminus \treeEdges$.
\end{lemma}
\begin{proof}
    In~$\Initialize$, every outgoing edge~$\edge=(\vertexA,\vertexB)\in\edges$ is added to~$\outEdges(\vertexA)$.
    If at some point~$\edge$ is not contained in~$\outEdges(\vertexA)$ but enters~$\residualContractedEdges$ as part of the set~$\revEdges$, then the reverse edge~$(\vertexB,\vertexA)$ is in~$\treeEdges$ and therefore admissible by Lemma~\ref{lem:tree}.
    Hence, $\edge$ itself is not admissible and does not need to be added to~$\outEdges(\vertexA)$.
    By Lemma~\ref{lem:distanceLabels:nonDecreasing}, distance labels never decrease, so if~$\edge$ is already contained in~$\residualContractedEdges$ and becomes admissible, this is because~$\distanceLabel(\vertexA)$ increases.
    This can only happen in~$\AdoptWithNewDist(\vertexA)$, where~$\outEdges(\vertexA)$ is rebuilt.
    Finally, $\edge$ is only removed from~$\outEdges(\vertexA)$ when it is examined in~$\AdoptWithSameDist(\vertexA)$.
    If~$\edge\in\residualContractedEdges$ and~$\edge$ is admissible, it is added to~$\treeEdges$, so the claim remains true.
\end{proof}

\begin{lemma} \label{lem:increaseDistanceLabel}
    If~$\AdoptWithNewDist(\vertexA)$ succeeds for a vertex~$\vertexA$, it strictly increases~$\distanceLabel(\vertexA)$.
\end{lemma}
\begin{proof}
    Line~\ref{alg:increaseDistanceLabel} sets~$\distanceLabel(\vertexA)$ to~$x:=\min_{(\vertexA,\vertexB)\in\residualContractedEdges}\distanceLabel(\vertexB)+1$.
    By Lemma~\ref{lem:distanceLabels:valid}, we know that~$\distanceLabel(\vertexA)\leq x$ holds before this step.
    Equality holds if and only if there is an admissible edge~$(\vertexA,\vertexB)\in\residualContractedEdges$.
    However, $\AdoptWithNewDist(\vertexA)$ is only called if~$\AdoptWithSameDist(\vertexA)$ fails, which means that~$\outEdges(\vertexA)$ does not contain any admissible edges.
    Because~$\vertexA$ is an orphan, it follows from Lemma~\ref{lem:tree} that there is no tree edge~$(\vertexA,\vertexB)\in\treeEdges$ either.
    Hence, by Lemma~\ref{lem:outedges}, there is no admissible outgoing edge in~$\residualContractedEdges$.
\end{proof}

\begin{lemma} \label{lem:treeAdditions}
    Each edge~$\edge=(\vertexA,\vertexB)$ is added to and removed from~$\tree$ at most~$n$ times.
\end{lemma}
\begin{proof}
    Between two subsequent additions of~$\edge$ to~$\tree$, it must be removed by~$\RemoveTreeEdge(\edge)$.
    There are two situations in which~$\RemoveTreeEdge(\edge)$ is called:
    \begin{itemize}
        \item Case 1: $\AdoptWithSameDist(\vertexB)$ fails, which means that~$\AdoptWithNewDist(\vertexB)$ is called afterwards.
        If this also fails, then~$\vertexB$ is contracted into~$\source$, so~$\edge$ cannot be re-added to~$\tree$.
        If it succeeds, then~$\distanceLabel(\vertexB)$ increases by Lemma~\ref{lem:increaseDistanceLabel}.
        \item Case 2: The edge~$\edge$ is part of the set~$\saturatedEdges$ of saturated tree edges in~$\UpdateResidualGraph$.
        In this case, it is also removed from~$\residualContractedEdges$.
        Because~$\edge$ is a tree edge at this point, it is admissible by Lemma~\ref{lem:tree}, i.e., $\distanceLabel(\vertexB)=\distanceLabel(\vertexA)-1$.
        To be re-added to~$\tree$, the edge must be re-added to~$\residualContractedEdges$ as part of~$\revEdges$ in~$\UpdateResidualGraph$ first.
        This requires that the reverse edge~$\edge'=(\vertexB,\vertexA)$ is in~$\treeEdges$ and therefore admissible, i.e., $\distanceLabel(\vertexB)=\distanceLabel(\vertexA)+1$.
        Because distance labels do not decrease (Lemma~\ref{lem:distanceLabels:nonDecreasing}), this can only be achieved by increasing~$\distanceLabel(\vertexB)$.
    \end{itemize}
    In both cases, $\distanceLabel(\vertexB)$ increases before~$\edge$ is re-added to~$\tree$.
    By Lemma~\ref{lem:distanceLabels:nonDecreasing}, this can happen at most~$n$ times.
\end{proof}

\begin{lemma}
\label{lem:adoptWithNewDist:fail}
    If~$\AdoptWithNewDist(\vertexA)$ fails for a vertex~$\vertexA$, there is no $\vertexA$-$\sink$-path in~$\residualContractedGraph$.
\end{lemma}
\begin{proof}
    If~$\AdoptWithNewDist(\vertexA)$ fails because there is no outgoing edge in~$\residualContractedGraph$, the claim is obviously true.
    Otherwise, it fails because~$\distanceLabel(\vertexC) \geq |\treeVertices|-1$.
    By Lemma~\ref{lem:distanceLabels:valid}, this is a lower bound for the length of all $\vertexC$-$\sink$-paths in~$\residualContractedGraph$.
    Because~$\vertexC$ was chosen such that~$\distanceLabel(\vertexC)$ is minimal, all~$\vertexA$-$\sink$-paths have length at least~$|\treeVertices|$.
    This implies that there are no simple~$\vertexA$-$\sink$-paths and therefore no paths at all.
\end{proof}

\subsection{Flow and Excess Functions}
\label{sec:proof:functions}
\begin{lemma}
    \label{lem:pseudoToFlow}
    At the start of every iteration, $\excess(\vertex)=0$ holds for every vertex~$\vertex\in\vertices_\sink\setminus\{\sink\}$, and for every parameter value~$\lambda$ such that~$\flow[\lambda]$ is a pseudoflow, $\flow[\lambda]$ is also a flow.
\end{lemma}
\begin{proof}
    We show that~$\excess(\vertex)(\lambda)$ is the excess of~$\vertex$ with respect to~$\flow[\lambda]$.
    At the start of~$\CalcFlowFunction$, all excesses are set to~$0$ and~$\flow$ is initialized with the constant flow function~$\initialFlow$, so~$\flow[\lambda]$ fulfills the flow conservation constraints for every parameter value~$\lambda$.
    Thus, if~$\flow[\lambda]$ is a pseudoflow, it is also a flow.
    There are three steps in the algorithm that change the flow along an edge~$\edge=(\vertexA,\vertexB)$: the initialization in~$\CalcFlowFunction$, $\RemoveTreeEdge(\edge)$ and~$\DrainExcess$.
    In all three cases, the difference between the old flow and the new flow is added to~$\excess(\vertexA)$ and subtracted from~$\excess(\vertexB)$ (unless the vertex in question is~$\source$ or~$\sink$).
    Whenever the excess function of a vertex~$\vertex$ is modified, the vertex is added to~$\excessVertices$.
    During~$\DrainExcess$, $\excess(\vertex)$ is set back to~0.
\end{proof}

\begin{lemma}
    \label{lem:forbiddenTreeEdges}
    Outgoing edges of~$\sink$ and incoming edges of~$\source$ are never contained in~$\treeEdges$.
\end{lemma}
\begin{proof}
    Because~$\distanceLabel(\source)=\infty$ holds initially, incoming edges of~$\source$ cannot become admissible and the distance label can never decrease.
    Likewise, because~$\distanceLabel(\sink)=0$ holds at all points during the execution of the algorithm, outgoing edges of~$\sink$ can never be admissible.
\end{proof}

Consider a non-tree edge~$\edge\in\edges\setminus\treeEdges$.
If the reverse edge is in~$\treeEdges$, we call~$\edge$ a~\emph{reverse tree edge}.
Otherwise, we call it an~\emph{out-of-tree edge}.

\begin{lemma}
    \label{lem:outOfTreeFlow}
    The algorithm maintains the following invariant for each out-of-tree edge~$\edge=(\vertexA,\vertexB)$ and its reverse edge~$\edge'=(\vertexB,\vertexA)$:
\begin{itemize}
    \item If~$\vertexA=\source$ or~$\vertexB=\sink$, then~$\residualCapacity(\edge)=0$ and~$\residualCapacity(\edge')=\capacity(\edge)+\capacity(\edge')$.
    \item If~$\vertexA,\vertexB\notin\{\source,\sink\}$, then~$\residualCapacity(\edge)$ and~$\residualCapacity(\edge')$ are constant.
\end{itemize}
\end{lemma}
\begin{proof}
    For the initial maximum flow~$\initialFlow/\vertices_\source$, edges of the form~$(\source,\vertexB)$ are saturated because they are part of a minimal cut.
    If a sink-incident edge~$(\vertexA,\sink)$ is residual, then it constitutes the only shortest~$\vertexA$-$\sink$-path in~$\residualContractedGraph$, so it must be a tree edge by Corollary~\ref{cor:tree}.
    Hence, all out-of-tree edges of the form~$(\source,\vertexB)$ or~$(\vertexA,\sink)$ are saturated for~$\initialFlow/\vertices_\source$, so~$\CalcFlowFunction$ initializes their flow function with the capacity function.
    For all other out-of-tree edges, $\CalcFlowFunction$ initializes their flow function with the constant initial flow.

    Afterwards, the flow functions of out-of-tree edges are never changed, so it suffices to show that the invariant holds if~$\edge$ is removed from~$\treeEdges$ via~$\RemoveTreeEdge(\edge)$.
    By Lemma~\ref{lem:forbiddenTreeEdges}, outgoing edges of~$\sink$ and incoming edges of~$\source$ never appear in~$\treeEdges$.
    Outgoing edges of~$\source$ are initially saturated and their reverse edges never appear in~$\treeEdges$ by Lemma~\ref{lem:forbiddenTreeEdges}, so they remain out-of-tree edges for the entire execution of the algorithm.
    Hence, there are only two cases to consider:
    If~$\vertexB\neq\sink$, then~$\newFlow$ is made constant.
    Because the capacities of~$\edge$ and~$\edge'$ are both constant, their residual capacities are as well.
    If~$\vertexB=\sink$, then it follows from Lemma~\ref{lem:forbiddenTreeEdges} that sink has no outgoing tree edges and is therefore not an orphan.
    Hence, $\RemoveTreeEdge(\edge)$ was called in~$\UpdateResidualGraph$, which means that~$\edge$ is saturated.
    Then~$\newFlow$ is set to~$\capacity(\edge)$, which yields~$\residualCapacity(\edge)=0$ and~$\residualCapacity(\edge')=\capacity(\edge)+\capacity(\edge')$.
\end{proof}

\begin{corollary}
    \label{cor:forbiddenTreeEdges}
    Outgoing edges of~$\source$ are never contained in~$\treeEdges$.
    Incoming edges of~$\sink$ are initially contained in~$\treeEdges$ if they are residual for~$\initialFlow$.
    If they are later removed from~$\treeEdges$, they are never re-added.
\end{corollary}
\begin{proof}
    This follows from Lemmas~\ref{lem:forbiddenTreeEdges} and~\ref{lem:outOfTreeFlow}.
\end{proof}

\begin{lemma}
    \label{lem:treeFlow}
    At the start of every iteration, the flow functions of all tree edges are non-decreasing.
    Accordingly, the residual capacity functions are non-increasing for tree edges and non-decreasing for reverse tree edges.
\end{lemma}
\begin{proof}
    For a vertex~$\vertexA\in\vertices_\sink\setminus\{\sink\}$, we define the \emph{net inflow function} of~$\vertexA$ as
    \[\netInflow(\vertexA):=\excess(\vertexA) + \sum_{(\vertexB,\vertexA)\in\treeEdges}\flow(\vertexB,\vertexA) + \sum_{(\vertexA,\vertexC)\in\treeEdges} \flow(\vertexC,\vertexA).\]
    We show that the algorithm maintains the invariant that~$\netInflow(\vertexA)$ is non-increasing.
    \texttt{CalcFlow-\\Function} ensures that~$\excess(\vertexA)=0$ and~$\sum_{(\vertexB,\vertexA)\in\edges}\flow(\vertexB,\vertexA)=0$.
    If we add these together and then subtract the flow functions of all incoming out-of-tree edges, this yields~$\netInflow(\vertexA)$.
    By Lemma~\ref{lem:outOfTreeFlow}, these flow functions are all non-decreasing, so~$\netInflow(\vertexA)$ is initially non-increasing.
    Whenever the flow function of an incoming edge~$(\vertexB,\vertexA)$ of~$\vertexA$ is changed, the difference between the new and old flow is subtracted from~$\excess(\vertexA)$, so~$\netInflow(\vertexA)$ remains unchanged.
    It remains to be shown that the invariant is upheld when an edge~$\edge=(\vertexA,\vertexB)$ is added to or removed from~$\treeEdges$.
    By Lemma~\ref{lem:forbiddenTreeEdges} and Corollary~\ref{cor:forbiddenTreeEdges}, we know that~$\vertexA\notin\{\source,\sink\}$ and~$\vertexB\neq\source$.
    When~$\edge$ is added during an adoption procedure for~$\vertexA$, then~$\vertexB\neq\sink$ also holds by Corollary~\ref{cor:forbiddenTreeEdges}.
    Thus, $\flow(\edge)$ is constant by Lemma~\ref{lem:outOfTreeFlow}, so~$\netInflow(\vertexA)$ and~$\netInflow(\vertexB)$ remain non-increasing when~$\flow(\edge)$ is added.
    When~$\edge$ is removed in~$\RemoveTreeEdge(\edge)$, there are two cases to consider:
    If~$\vertexB=\sink$, then~$\capacity(\edge)$ is non-decreasing, so~$\flow(\edge)$ is also made non-decreasing.
    This non-decreasing function is subtracted from $\netInflow(\vertexA)$, which remains non-increasing.
    The addition to~$\netInflow(\sink)$ is irrelevant.
    If~$\vertexB\neq\sink$, then~$\flow(\edge)$ is made constant, so~$\netInflow(\vertexA)$ and~$\netInflow(\vertexB)$ remain non-increasing.

    Consider a tree edge~$\edge=(\vertexA,\vertexC)\in\treeEdges$.
    At the start of every iteration, $\excess(\vertexA)=0$ by Lemma~\ref{lem:pseudoToFlow}.
    Because~$\edge$ is the only outgoing tree edge of~$\vertexA$, it follows that~$\netInflow(\vertexA)=\sum_{(\vertexB,\vertexA)\in\treeEdges}\flow(\vertexB,\vertexA) - \flow(\vertexA,\vertexC)$.
    This function is non-increasing, so~$\flow(\vertexA,\vertexC) - \sum_{(\vertexB,\vertexA)\in\treeEdges}\flow(\vertexB,\vertexA)=-\netInflow(\vertexA)$ is non-decreasing.
    Then it follows by induction from the leaves upwards that the flow functions of all tree edges are non-decreasing.
    Because the capacities are non-increasing, the residual capacities are as well.
\end{proof}

\begin{lemma}
    \label{lem:rootLambda}
    Assume that~$\flow[\lambda_i]$ is a flow in~$\contractedGraph$ at the start of iteration~$i$.
    Then the following holds for every edge~$\edge=(\vertexA,\vertexB)\in\contractedEdges$:
    If~$\residualCapacity(\edge)$ is non-negative for all parameter values in~$[\lambda_i,\maxLambda]$, then~$\rootLambda(\edge) \geq \maxLambda$.
    Otherwise, $\rootLambda(\edge)$ is the smallest root of~$\residualCapacity(\edge)$ in~$[\lambda_i,\maxLambda]$.
\end{lemma}
\begin{proof}
    Whenever~$\edge$ is not in~$\treeEdges$, the algorithm ensures that~$\rootLambda(\edge)=\infty$ holds.
    In this case, we show that~$\residualCapacity(\edge)(\lambda) \geq 0$ holds for~$\lambda\in[\lambda_i,\maxLambda]$.
    Note that if~$\residualCapacity(\edge)$ is non-decreasing, this follows because~$\flow[\lambda_i]$ is a flow and therefore~$\residualCapacity(\edge)(\lambda_i) \geq 0$ holds.
    If~$\edge$ is a reverse tree edge, $\residualCapacity(\edge)$ is non-decreasing by Lemma~\ref{lem:treeFlow}.
    If~$\edge$ is an out-of-tree edge, there are several cases to consider by Lemma~\ref{lem:outOfTreeFlow}:
    If~$\vertexA\neq\sink$ and~$\vertexB\neq\source$, then~$\residualCapacity(\edge)$ is constant and therefore non-decreasing.
    Otherwise, $\residualCapacity(\edge)(\lambda)=\capacity(\edge)(\lambda)+\capacity(\edge')(\lambda)$ holds, which is non-negative for every parameter value~$\lambda\in[\minLambda,\maxLambda]$.

    It remains to be shown that~$\rootLambda(\edge)$ is correct at the start of every iteration in which~$\edge\in\treeEdges$ holds.
    For the first iteration, this is the case because~$\rootLambda(\edge)$ is calculated via~$\CalcRoot$ in~$\CalcFlowFunction$.
    If~$\edge$ enters~$\treeEdges$ via an adoption procedure, then~$\vertexA\in\orphans\subseteq\excessVertices$ holds.
    This means that~$\rootLambda(\edge)$ is recalculated via~$\CalcRoot$ when~$\edge$ is processed in~$\DrainExcess$.
    $\CalcRoot$ returns the smallest root in~$[\lambda_i,\infty)$ if there is one.
    If this value is greater than~$\maxLambda$, then~$\residualCapacity(\edge)$ is non-negative in~$[\lambda_i,\maxLambda]$.
\end{proof}

\subsection{Extending the Flow}
\label{sec:proof:extend}
The algorithm maintains the invariant that~$\flow[\lambda_i]$ is a maximum flow in~$\contractedGraph$ with residual graph~$\residualContractedGraph$ at the start of iteration~$i$.
As~$\lambda$ increases, the flow~$\flow[\lambda]$ and the capacities of the contracted graph~$\contractedGraph[\lambda]:=\graph[\lambda]/\vertices_\source$ change.
We show that the corresponding residual graph changes to
\[\residualContractedGraph[\lambda]:=(\contractedVertices,\residualContractedEdges\cup\revEdges[\lambda]\setminus\saturatedEdges[\lambda],\residualCapacity[\lambda],\source,\sink),\]
where the set of previously saturated edges that become residual is given by
\begin{align*}
\revEdges[\lambda] := \{(\vertexB,\vertexA)\mid (\vertexA,\vertexB)\in\treeEdges \land (\vertexB,\vertexA) \notin \residualContractedEdges\land \residualCapacity[\lambda](\vertexB,\vertexA) > 0\}
\end{align*}
and the set of previously residual edges that become saturated is given by
\begin{align*}
\saturatedEdges[\lambda] :&= \{ \edge \in \treeEdges \mid \residualCapacity[\lambda](\edge) = 0 \} = \{ \edge \in \treeEdges \mid \rootLambda(\edge) = \lambda \}.
\end{align*}

\begin{lemma}
    \label{lem:extendFlow}
    Assume that~$\flow[\lambda_i]$ is a maximum flow in~$\contractedGraph$ with residual graph~$\residualContractedGraph$ at the start of iteration~$i$.
    Then for every parameter value~$\lambda \in [\lambda_i,\min(\lambda_{i+1},\maxLambda)]$, $\flow[\lambda]$ is a maximum flow in~$\contractedGraph[\lambda]$ with residual graph~$\residualContractedGraph[\lambda]$.
\end{lemma}
\begin{proof}
    The algorithm ensures that~$\rootLambda(\edge)=\infty$ holds for every edge~$\edge \in \contractedEdges \setminus \treeEdges$.
    Because~$\lambda_{i+1}$ is calculated as~$\min_{\edge\in\treeEdges}\rootLambda(\edge)$, it follows from Lemma~\ref{lem:rootLambda} that all residual capacities are non-negative for $\lambda\in[\lambda_i,\min(\lambda_{i+1},\maxLambda)]$.
    Hence, all capacity constraints are fulfilled.
    Furthermore, $\SetFlow$ ensures that the antisymmetry constraints are upheld, so~$\flow[\lambda]$ is a pseudoflow.
    Then it is also a flow by Lemma~\ref{lem:pseudoToFlow} and maximal because it saturates the cut~$(\{\source\}, \vertices_\sink)$ by Lemma~\ref{lem:outOfTreeFlow}.

    To show that the residual graph is~$\residualContractedGraph[\lambda]$, consider an edge~$\edge=(\vertexA,\vertexB)\in\edges$.
    By Lemma~\ref{lem:outOfTreeFlow}, $\residualCapacity[\lambda](\edge)$ can only differ from~$\residualCapacity[\lambda_i](\edge)$ if~$\edge$ is a tree edge, reverse tree edge or out-of-tree edge with~$\vertexA=\sink$ or~$\vertexB=\source$ for~$\lambda_i$.
    If~$\edge$ is a tree edge, then it is residual for~$\lambda_i$ by Lemma~\ref{lem:tree}.
    By Lemma~\ref{lem:rootLambda}, it becomes saturated if~$\rootLambda(\edge)=\lambda$.
    If~$\edge$ is a reverse tree edge, the residual capacity is non-decreasing, as shown in Lemma~\ref{lem:rootLambda}.
    Hence, if it is residual for~$\lambda_i$, it remains residual for~$\lambda$.
    If it is saturated for~$\lambda_i$ and becomes residual for~$\lambda$, it is contained in~$\revEdges[\lambda]$.
    If~$\edge$ is an out-of-tree edge with~$\vertexA=\sink$ or~$\vertexB=\source$, then the reverse edge is saturated for both~$\lambda$ and~$\lambda_i$ by Lemma~\ref{lem:outOfTreeFlow}, so~$\edge$ remains residual.
\end{proof}

\begin{corollary} \label{cor:nextParameter}
    The new parameter value~$\lambda_{i+1}$ is strictly greater than the previous parameter value~$\lambda_i$.
\end{corollary}
\begin{proof}
    Because $\lambda_{i+1}$ is chosen as~$\min_{\edge \in \treeEdges} \rootLambda(\edge)$, the set~$\saturatedEdges[\lambda_{i+1}]$ is non-empty.
    However, $\saturatedEdges[\lambda_i]$ is empty because all edges in~$\tree$ are in~$\residualContractedGraph$ by Lemma~\ref{lem:tree} and therefore residual by Invariant~\ref{inv:flow}.
    This implies~$\lambda_i\neq\lambda_{i+1}$.
\end{proof}

\subsection{Correctness}
\label{sec:proof:correctness}
For any variable~$x$ used in the algorithm, we denote the value of~$x$ at the start of iteration~$i$ by~$x^{(i)}$.
We show that the algorithm upholds the following invariants at the start of each iteration~$i$:
\begin{enumerate}
    \item \label{inv:breakpoint} $\breakpoint^{(i)}$ is a breakpoint function for~$\graph$ in the interval~$[\minLambda,\lambda_i]$ with~$\cut(\breakpoint^{(i)},\lambda_i)=\cut^{(i)}$.
    \item \label{inv:flow} $\flow^{(i)}[\lambda_i]$ is a maximum flow in $\contractedGraph^{(i)}$ whose set of residual edges is~$\residualContractedEdges^{(i)}$.
\end{enumerate}

\begin{lemma} \label{lem:reconnectTree:paths}
    Assume that both invariants hold at the start of iteration~$i$.
    Then~\normalfont{\texttt{Reconnect-\\Tree}} terminates and yields~$\orphans=\emptyset$, $\vertices_\sink^{(i+1)}=\vertices_\sink^{(i)}\setminus\newVertices$ and~$\residualContractedGraph^{(i+1)}=\residualContractedGraph^{(i)}[\lambda_{i+1}]/\newVertices$.
\end{lemma}
\begin{proof}
    $\UpdateResidualGraph$ ensures that~$\residualContractedGraph=\residualContractedGraph^{(i)}[\lambda_{i+1}]$.
    If both adoption procedures fail for an orphan~$\vertex \in \orphans$, it is added to~$\newVertices$ and contracted into~$\source$, so it cannot become an orphan again.
    If either procedure succeeds, an edge is added to~$\tree$, which can only happen finitely many times by Lemma~\ref{lem:treeAdditions}.
    Hence, $\orphans$ becomes empty after a finite number of iterations and~$\ReconnectTree$ terminates with exactly the vertices in~$\newVertices$ contracted into~$\source$.
\end{proof}

\begin{lemma}
    \label{lem:cut}
    Assume that both invariants hold at the start of iteration~$i$.
    Then the sink-minimal minimum cut in~$\graph[\lambda]$ is~$\cut^{(i)}$ for~$\lambda\in[\lambda_i,\lambda_{i+1})$ and~$\cut^{(i+1)}$ for~$\lambda_{i+1}$.
\end{lemma}
\begin{proof}
    Due to Lemma~\ref{lem:contract}, it suffices to show that the sink-minimal minimum cut in~$\contractedGraph^{(i)}[\lambda]$ is~$(\{ \source \}, \vertices_\sink^{(i)})$ for~$\lambda\in[\lambda_i,\lambda_{i+1})$ and~$(\{\source\} \cup \newVertices, \vertices_\sink^{(i+1)})$ for~$\lambda_{i+1}$.
    By Lemma~\ref{lem:extendFlow}, $\flow^{(i)}[\lambda]$ is a maximum flow in~$\contractedGraph^{(i)}[\lambda]$ with residual graph~$\residualContractedGraph^{(i)}[\lambda]$.
    We show that the set of vertices from which~$\sink$ is reachable in~$\residualContractedGraph^{(i)}[\lambda]$ is~$\vertices_\sink^{(i)}$ for~$\lambda\in[\lambda_i,\lambda_{i+1})$ and~$\vertices_\sink^{(i+1)}$ for~$\lambda_{i+1}$.
    For~$\lambda=\lambda_i$, this already follows from Invariant~\ref{inv:breakpoint}.
    For~$\lambda>\lambda_i$, the residual graph changes via the addition of the newly residual edges~$\revEdges[\lambda]$ and the removal of the newly saturated edges~$\saturatedEdges[\lambda]$.
    Because~$\revEdges[\lambda]$ does not contain any edges incident to~$\source$, no~$\source$-$\sink$-path is added.    
    For~$\lambda<\lambda_{i+1}$, $\saturatedEdges[\lambda]$ is empty, so all existing paths are preserved and~$\sink$ remains reachable from all vertices in~$\vertices_\sink^{(i)}$.
    For~$\lambda=\lambda_{i+1}$, it follows from Lemma~\ref{lem:reconnectTree:paths} and Corollary~\ref{cor:tree} that~$\tree^{(i+1)}$ is a shortest-path-tree to~$\sink$ in~$\residualContractedGraph^{(i)}[\lambda]$.
    Hence, all vertices in~$\vertices_\sink^{(i+1)}$ have a path to~$\sink$ in~$\residualContractedGraph^{(i)}[\lambda_{i+1}]$.
    For each vertex~$\vertex\in\newVertices$ that was removed from~$\tree$, $\AdoptWithNewDist(\vertex)$ failed, so by Lemma~\ref{lem:adoptWithNewDist:fail} there is no~$\vertex$-$\sink$-path in~$\residualContractedGraph^{(i)}[\lambda_{i+1}]$.
\end{proof}

\begin{lemma} \label{lem:flowChange}
    At the end of iteration~$i$, it holds that~$\flow^{(i+1)}[\lambda_{i+1}]=\flow^{(i)}[\lambda_{i+1}]/\newVertices$.
\end{lemma}
\begin{proof}
    At the start of the iteration, all excess functions are 0 by Lemma~\ref{lem:pseudoToFlow}.
    The parts of the iteration that modify flow values and excess functions are~$\RemoveTreeEdge$ and~\texttt{Drain-}\\$\texttt{Excess}(\lambda_{i+1})$.
    Consider a call to~$\RemoveTreeEdge(\edge)$ for some edge~$\edge=(\vertexA,\vertexB)$ with reverse edge~$\edge'=(\vertexB,\vertexA)$.
    Because~$\newFlow[\lambda_{i+1}] = \flow(\edge)[\lambda_{i+1}]$ holds for the new flow function~$\newFlow$ by construction, $\flow(\edge)[\lambda_{i+1}]$ and~$\flow(\edge')[\lambda_{i+1}]$ remain unchanged.
    The excesses of~$\vertexA$ and~$\vertexB$ are modified by adding and subtracting~$\flow(\edge)-\newFlow$, respectively, which leaves~$\excess(\vertexA)[\lambda_{i+1}]=\excess(\vertexB)[\lambda_{i+1}]=0$ unchanged.
    If~$\edge$ is processed in~$\DrainExcess(\lambda_{i+1})$, then~$\excess(\vertexA)$ is added to~$\flow(\edge)$ and subtracted from~$\flow(\edge')$.
    Additionally, $\excess(\vertexA)$ is added to~$\excess(\vertexB)$ and~$\excess(\vertexA)$ is set to 0.
    This upholds~$\excess(\vertexA)[\lambda_{i+1}]=\excess(\vertexB)[\lambda_{i+1}]=0$ and therefore also leaves~$\flow(\edge)[\lambda_{i+1}]$ and~$\flow(\edge')[\lambda_{i+1}]$ unchanged.
    Because~$\residualContractedGraph^{(i+1)}[\lambda_{i+1}]=\residualContractedGraph^{(i)}[\lambda_{i+1}]/\newVertices$ holds at the end of the iteration by Lemma~\ref{lem:reconnectTree:paths}, the claim follows.
\end{proof}

\begin{lemma} \label{lem:invariants}
    If both invariants hold at the start of iteration~$i$ and the algorithm does not terminate during iteration~$i$, they also hold at the start of iteration~$i+1$, i.e.,
    \begin{enumerate}
        \item $\breakpoint^{(i+1)}$ is a breakpoint function for~$\graph$ in the interval~$[\minLambda,\lambda_{i+1}]$ with~$\cut(\breakpoint^{(i+1)},\lambda_{i+1})=\cut^{(i+1)}$.
        \item $\flow^{(i+1)}[\lambda_{i+1}]$ is a maximum flow in~$\contractedGraph^{(i+1)}$ with residual graph~$\residualContractedGraph^{(i+1)}$.
    \end{enumerate}
\end{lemma}
\begin{proof}
    With Lemma~\ref{lem:reconnectTree:paths}, it follows that
    \[
    \cut(\breakpoint^{(i+1)},\lambda)=\begin{cases}
    \cut(\breakpoint^{(i)},\lambda) & \text{if } \lambda\in[\minLambda,\lambda_i),\\
    \cut^{(i)} & \text{if } \lambda\in[\lambda_i,\lambda_{i+1}),\\
    \cut^{(i+1)} & \text{otherwise.}\\
    \end{cases}
    \]
    For~$\lambda\in[\minLambda,\lambda_i)$, it follows from Invariant~\ref{inv:breakpoint} that this is the sink-minimal minimum cut in~$\graph[\lambda]$.
    For the other two cases, it follows from Lemma~\ref{lem:cut}, so Invariant~\ref{inv:breakpoint} is upheld.

    By Lemma~\ref{lem:extendFlow}, $\flow^{(i)}[\lambda_{i+1}]$ is a maximum flow in~$\contractedGraph^{(i)}[\lambda_{i+1}]$ with residual graph~$\residualContractedGraph^{(i)}[\lambda_{i+1}]$.
    By Lemma~\ref{lem:reconnectTree:paths}, $\newVertices$ is fully contained in the source component of~$\cut^{(i+1)}$, which is a minimum cut in~$\graph[\lambda_{i+1}]$ by Lemma~\ref{lem:cut}.
    Hence, by Lemma~\ref{lem:contractFlow}, $\flow^{(i)}[\lambda_{i+1}]/\newVertices$ is a maximum flow in~$\contractedGraph^{(i)}[\lambda_{i+1}]/\newVertices$ with residual graph~$\residualContractedGraph^{(i)}[\lambda_{i+1}]/\newVertices$.
    By Lemma~\ref{lem:reconnectTree:paths}, these graphs equal~$\contractedGraph^{(i+1)}$ and~$\residualContractedGraph^{(i+1)}$, respectively.
    By Lemma~\ref{lem:flowChange}, $\flow^{(i)}[\lambda_{i+1}]/\newVertices=\flow^{(i+1)}[\lambda_{i+1}]$.
\end{proof}

\begin{theorem}
    \PBFS returns a breakpoint function.
\end{theorem}
\begin{proof}
    Invariant~\ref{inv:breakpoint} holds at the start of the first iteration because~$\Initialize$ ensures that~$\breakpoint^{(0)}$ is a breakpoint function for~$\graph$ in~$[\minLambda,\lambda_0]$ with~$\cut(\breakpoint^{(0)},\lambda_0)=(\vertices_\source^{(0)},\vertices_\sink^{(0)})$.
    The flow~$\initialFlow$ calculated in~$\Initialize$ is a maximum flow in~$\graph[\minLambda]$ with residual graph~$\residualGraph$.
    By Lemma~\ref{lem:contractFlow}, $\flow^{(0)}[\lambda_0]=\initialFlow/\vertices_\source^{(0)}$ is a maximum flow in~$\contractedGraph^{(0)}=\graph[\minLambda]/\vertices_\source^{(0)}$ with residual graph~$\residualContractedGraph^{(0)}$, so Invariant~\ref{inv:flow} also holds.
    For each subsequent iteration, Lemma~\ref{lem:invariants} shows that the invariants are upheld.
    Because~$\lambda_i$ strictly increases with each iteration according to Corollary~\ref{cor:nextParameter}, the algorithm terminates either because the loop condition~$\lambda_i<\maxLambda$ fails or the return condition~$\lambda_{i+1}>\maxLambda$ in line~\ref{alg:return} is met.
    In the former case, the claim follows immediately from Invariant~\ref{inv:breakpoint}.
    In the latter case, it follows from Invariant~\ref{inv:breakpoint} and Lemma~\ref{lem:cut}.
\end{proof}

\subsection{Runtime}
\label{sec:proof:runtime}
\begin{lemma} \label{lem:reconnectTime}
    The time spent in~$\ReconnectTree$ throughout the algorithm is in~$\mathcal{O}(nm)$.
\end{lemma}
\begin{proof}
    The runtime of~$\ReconnectTree$ is dominated by the runtimes of~$\AdoptWithSameDist$ and~$\AdoptWithNewDist$.
    For a vertex~$\vertex$, the time spent in~$\AdoptWithSameDist(\vertex)$ is proportional to the number of edges that are removed from~$\outEdges(\vertex)$.
    Each time a removed edge~$\edge$ is re-added to~$\outEdges(\vertex)$ in~$\AdoptWithNewDist(\vertexA)$, the distance label of~$\vertexA$ increases by Lemma~\ref{lem:increaseDistanceLabel}.
    By Lemma~\ref{lem:distanceLabels:nonDecreasing}, this can happen at most~$n$ times, which yields an overall runtime of~$\mathcal{O}(nm)$ for~$\AdoptWithSameDist$.
    Each time~$\AdoptWithNewDist(\vertex)$ is called, it either increases~$\distanceLabel(\vertex)$ or removes~$\vertex$ from~$\residualContractedGraph$, in which case~$\AdoptWithNewDist(\vertex)$ is never called again.
    Because the time for a single call is in~$\mathcal{O}(\deg(\vertex))$ and~$\distanceLabel(\vertex)$ can increase at most~$n$ times, the overall time for~$\AdoptWithNewDist$ is also in~$\mathcal{O}(nm)$.
\end{proof}

\begin{theorem}
    \PBFS runs in time $\mathcal{O}(n^2m)$.
\end{theorem}
\begin{proof}
    $\CalcMaxFlow$ runs in time~$\mathcal{O}(n^2m)$ when performed by IBFS.
    The remaining steps in~$\Initialize$, as well as~$\CalcFlowFunction$, require time~$\mathcal{O}(n + m)$.
    If~$\lambda_{i+1}>\maxLambda$, the algorithm terminates.
    Otherwise, the following~$\UpdateResidualGraph$ call removes at least one saturated edge from~$\tree$.
    By Lemma~\ref{lem:treeAdditions}, each edge can be removed at most~$\mathcal{O}(n)$ times, so the total number of iterations of the while-loop is in~$\mathcal{O}(nm)$.
    Each~$\UpdateResidualGraph$ and~$\DrainExcess$ call takes time~$\mathcal{O}(n)$, so the total time spent in these calls is in~$\mathcal{O}(n^2m)$ time.
    The overall time for~$\ReconnectTree$ is in~$\mathcal{O}(nm)$ by Lemma~\ref{lem:reconnectTime}.
    In total, we obtain the runtime bound of $\mathcal{O}(n^2m)$.
\end{proof}

\begin{table*}[h]
    \caption{%
        Benchmark instances used in our experiments.
        The number of vertices and edges are denoted by~$n$ and~$m$, respectively.
        The polygon aggregation instances are listed first.
        The first group of instances, from \texttt{ahrem} to \texttt{groebzig}, consists of small villages.
        The following group, from \texttt{aldenhoven} to \texttt{forsbach}, are towns.
        The \texttt{group} instances are larger regions encompassing multiple villages and town.
        \texttt{bonn} and \texttt{cologne} are cities, whereas \texttt{big\_set} is a very large region.
        The last group consists of the image segmentation instances.
    }%
    \label{tbl:instances}%
    \begin{minipage}[t]{0.46\textwidth}
    \small{
    \begin{tabular}[t]{l@{\hskip 1.8cm}rr}
        \toprule
        Instance & $n$ & $m$ \\
        \midrule
        \texttt{ahrem}  &  5\,508  &  39\,927\\
        \texttt{altenholz}  &  15\,862  &  114\,179\\
        \texttt{altenrath}  &  8\,592  &  61\,973\\
        \texttt{belm}  &  13\,704  &  98\,029\\
        \texttt{berga}  &  7\,051  &  50\,979\\
        \texttt{bockelskamp}  &  2\,980  &  21\,295\\
        \texttt{bokeloh}  &  4\,842  &  34\,763\\
        \texttt{braunlage}  &  10\,983  &  78\,873\\
        \texttt{buisdorf}  &  10\,589  &  76\,789\\
        \texttt{butzweiler}  &  4\,950  &  35\,503\\
        \texttt{duengenheim}  &  7\,578  &  54\,947\\
        \texttt{edendorf}  &  17\,408  &  125\,957\\
        \texttt{erlenbach}  &  6\,179  &  44\,531\\
        \texttt{erp}  &  13\,269  &  96\,137\\
        \texttt{friesheim}  &  8\,712  &  62\,649\\
        \texttt{gerolstein}  &  11\,556  &  82\,673\\
        \texttt{gevenich}  &  4\,629  &  33\,407\\
        \texttt{gluecksburg}  &  10\,050  &  72\,381\\
        \texttt{goddula}  &  6\,497  &  46\,517\\
        \texttt{goldenstedt}  &  12\,139  &  86\,869\\
        \texttt{groebzig}  &  7\,370  &  53\,505\\[5pt]
        \texttt{aldenhoven}  &  28\,551  &  207\,717\\
        \texttt{andernach}  &  38\,533  &  279\,943\\
        \texttt{bad\_harzburg}  &  22\,163  &  158\,799\\
        \texttt{bad\_neunahr}  &  38\,001  &  273\,933\\
        \texttt{bergedorf}  &  67\,513  &  489\,635\\
        \texttt{celle}  &  103\,297  &  747\,683\\
        \texttt{euskirchen}  &  100\,379  &  736\,905\\
        \texttt{forsbach}  &  23\,745  &  170\,831\\
        \bottomrule
    \end{tabular}
    }
    \end{minipage}
    \hspace{0.03\textwidth}
    \begin{minipage}[t]{0.48\textwidth}
    \small{
    \begin{tabular}[t]{lrr}
        \toprule
        Instance & $n$ & $m$ \\
        \midrule
        \texttt{gruppe1}  &  39\,904  &  285\,711\\
        \texttt{gruppe2}  &  23\,137  &  166\,229\\
        \texttt{gruppe3}  &  29\,553  &  213\,107\\
        \texttt{gruppe4}  &  40\,422  &  289\,887\\
        \texttt{gruppe5}  &  46\,920  &  338\,135\\
        \texttt{gruppe6}  &  46\,270  &  331\,557\\
        \texttt{gruppe7}  &  53\,956  &  386\,839\\
        \texttt{gruppe8}  &  22\,143  &  158\,687\\
        \texttt{gruppe9}  &  37\,533  &  269\,163\\
        \texttt{gruppe10}  &  86\,491  &  624\,025\\
        \texttt{gruppe11}  &  70\,486  &  510\,565\\
        \texttt{gruppe12}  &  76\,316  &  549\,311\\[5pt]
        \texttt{bonn}  &  518\,449  &  3\,743\,553\\
        \texttt{cologne}  &  1\,135\,846  &  8\,222\,807\\
        \texttt{big\_set}  &  3\,383\,359  &  24\,588\,269\\[5pt]
        \texttt{adhead.n6c10} &  12\,582\,914 &  125\,829\,122\\
        \texttt{adhead.n6c100} &  12\,582\,914 &  125\,829\,122\\
        \texttt{adhead.n26c10} &  12\,582\,914 &  377\,487\,362\\
        \texttt{adhead.n26c100} &  12\,582\,914 &  377\,487\,362\\
        \texttt{babyface.n6c10} &  5\,062\,502 &  50\,625\,002\\
        \texttt{babyface.n6c100} &  5\,062\,502 &  50\,625\,002\\
        \texttt{babyface.n26c10} &  5\,062\,502 &  151\,875\,002\\
        \texttt{babyface.n26c100} &  5\,062\,502 &  151\,875\,002\\
        \texttt{bone.n6c10} &  7\,798\,786 &  77\,987\,842\\
        \texttt{bone.n6c100} &  7\,798\,786 &  77\,987\,842\\
        \texttt{bone.n26c10} &  7\,798\,786 &  233\,963\,522\\
        \texttt{bone.n26c100} &  7\,798\,786 &  233\,963\,522\\
        \texttt{liver.n6c10} &  4\,161\,602 &  41\,616\,002\\
        \texttt{liver.n6c100} &  4\,161\,602 &  41\,616\,002\\
        \texttt{liver.n26c10} &  4\,161\,602 &  124\,848\,002\\
        \texttt{liver.n26c100} &  4\,161\,602 &  124\,848\,002\\
        \bottomrule
    \end{tabular}
    }
    \end{minipage}
\end{table*}

\section{Experiments}
\label{sec:experiments}
We evaluate the practical performance of \PBFS by comparing it to \dichotomicscheme on instances from two classes of applications: polygon aggregation and image segmentation.
For the static max-flow algorithm used by \dichotomicscheme, we test both PRF and IBFS.
We do not compare to GGT because the results by Babenko et al.~\cite{BDGTZ07} suggest that it is outperformed by~\dichotomicscheme unless the splits are extremely imbalanced; we show that this is not the case for any of the tested instances.
To ensure that all algorithms run on the same graph data structure, we used our own C\raisebox{0.15ex}{\small++} implementations of all algorithms.
For PRF and IBFS, we closely follow the public implementations by Cherkassky and Goldberg~\cite{CG97} and Goldberg et al.~\cite{GHKKTW15}, respectively.
Our graph data structure requires that every edge has a reverse edge and that every vertex has a source-incident and sink-incident edge; missing edges are added with capacity~0 when the graph is loaded.
All experiments were run on a single thread of a machine with an AMD EPYC 7543P CPU, 32 cores and~256\,GB RAM.
The code was compiled with GCC 11.4.0\footnote{The source code is available at \url{https://github.com/jonas-sauer/MonotoneParametricMinCut}.}.

\subparagraph*{Instances.}
To evaluate the algorithms for the task of polygon aggregation, we created a new set of~44 benchmark instances\footnote{The benchmark set is made publicly available at~\url{https://doi.org/10.5281/zenodo.13642985}.}.
These consist of building footprints from municipalities and regions in Western Germany, which were extracted from OpenStreetMap and processed as described by Rottmann et al.~\cite{RDGRH21}.
To the best of our knowledge, none of the public benchmark datasets from computer vision include monotone parametric instances.
Therefore, we add synthetic parametric capacities to a 3D image segmentation dataset for the static max-flow problem\footnote{Downloaded from~\url{https://vision.cs.uwaterloo.ca/data/maxflow}.}.
This dataset consists of regular grid graphs derived from medical scans.
Each scan is turned into four instances, where instance \texttt{imagename.nxcy} has $x$ neighbors per vertex (excluding~$\source$ and~$\sink$) and a maximum edge capacity of~$y$.
For each edge~$\edge=(\vertexA,\vertexB)$ with original static capacity~$\capacity_0(\edge)$, we assign a parametric capacity function~$\capacity(\edge)$ as follows.
If~$\vertexA=\source$, we set~$\capacity(\edge)(\lambda)=\capacity_1 + \lambda \cdot \capacity_2$, where~$\capacity_1$ and~$\capacity_2$ are chosen uniformly at random from~$[1,y]$.
Otherwise, we set~$\capacity(\edge)(\lambda)=\capacity_0(\edge)$.
An overview of the instances is given in Table~\ref{tbl:instances}.

\subparagraph*{Performance.}
\begin{table*}
    \center
    \caption{%
        Performance of \PBFS and \dichotomicscheme-IBFS on polygon aggregation instances.
        For each algorithm, Time is the runtime in milliseconds and BP is the number of computed breakpoints.
        For \PBFS, Ad/BP is the number of successful adoptions, divided by the number of breakpoints, and Loop/Init is the time spent in the main loop, divided by the initialization time.
        For \dichotomicscheme, Vert is the combined number of vertices of all contracted graphs built by the algorithm, divided by~$n$, and Contr is the percentage of the runtime that is spent building contracted graphs.
    }%
    \label{tbl:runtime-full-agg}%
    \small{
    \begin{tabular*}{\textwidth}{@{\,}l@{\extracolsep{\fill}}r@{\extracolsep{\fill}}r@{\extracolsep{\fill}}r@{\extracolsep{\fill}}r@{\extracolsep{\fill}}r@{\extracolsep{\fill}}r@{\extracolsep{\fill}}r@{\extracolsep{\fill}}r@{\extracolsep{\fill}}r@{\,}}
        \toprule
        \multirow{2}{*}{Instance} & \multicolumn{4}{c}{\PBFS} & \multicolumn{4}{c}{Dichotomic with IBFS} & \multirow{2}{*}{Speedup} \\
        \cmidrule{2-5} \cmidrule{6-9}
        & Time & BP & Ad./BP & Loop/Init & Time & BP & Vert. & Contr. \\
        \midrule
        \texttt{ahrem}  & 11.9 &  1\,891  & 20.1 & 10.33 & 35.1 &  1\,890  & 12.78 &  41\%  &  2.95\\
        \texttt{altenholz}  & 48.4 &  3\,342  & 52.9 & 16.36 & 109.9 &  3\,335  & 11.98 &  34\%  &  2.27\\
        \texttt{altenrath}  & 22.1 &  2\,742  & 26.3 & 11.54 & 59.2 &  2\,740  & 12.74 &  39\%  &  2.68\\
        \texttt{belm}  & 36.9 &  2\,421  & 54.3 & 14.95 & 85.3 &  2\,422  & 11.28 &  36\%  &  2.31\\
        \texttt{berga}  & 16 &  1\,977  & 25.3 & 10.94 & 43.5 &  1\,976  & 12.57 &  42\%  &  2.71\\
        \texttt{bockelskamp}  & 6.3 & 859 & 25.9 & 10.21 & 16.7 & 859 & 10.59 &  41\%  &  2.64\\
        \texttt{bokeloh}  & 11.5 &  1\,036  & 41.6 & 11.95 & 27.9 &  1\,034  & 10.52 &  37\%  &  2.43\\
        \texttt{braunlage}  & 28.9 &  2\,383  & 41.9 & 13.24 & 69 &  2\,383  & 11.7 &  38\%  &  2.39\\
        \texttt{buisdorf}  & 27 &  2\,991  & 28.8 & 12.96 & 68.5 &  2\,985  & 12.54 &  39\%  &  2.54\\
        \texttt{butzweiler}  & 12.6 &  1\,347  & 32.7 & 10.39 & 29.2 &  1\,348  & 10.95 &  38\%  &  2.31\\
        \texttt{duengenheim}  & 18 &  2\,581  & 22.7 & 11.29 & 50 &  2\,575  & 13.03 &  40\%  &  2.78\\
        \texttt{edendorf}  & 51.4 &  4\,953  & 33.6 & 15.14 & 124.6 &  4\,947  & 13.3 &  37\%  &  2.42\\
        \texttt{erlenbach}  & 14.1 &  1\,624  & 30.2 & 11.56 & 37.2 &  1\,623  & 11.5 &  38\%  &  2.63\\
        \texttt{erp}  & 32.8 &  4\,663  & 21 & 11.85 & 91.4 &  4\,649  & 13.82 &  41\%  &  2.79\\
        \texttt{friesheim}  & 22.5 &  2\,387  & 32.2 & 12.23 & 55.1 &  2\,382  & 11.98 &  38\%  &  2.45\\
        \texttt{gerolstein}  & 31.4 &  2\,577  & 41.3 & 13.28 & 73.2 &  2\,574  & 11.86 &  37\%  &  2.33\\
        \texttt{gevenich}  & 10.5 &  1\,426  & 24.1 & 10.24 & 28.6 &  1\,424  & 12.82 &  40\%  &  2.72\\
        \texttt{gluecksburg}  & 26.2 &  2\,741  & 32.7 & 12.77 & 65.7 &  2\,741  & 12.33 &  38\%  &  2.51\\
        \texttt{goddula}  & 15 &  1\,758  & 28.6 & 11.11 & 39.5 &  1\,759  & 11.72 &  40\%  &  2.64\\
        \texttt{goldenstedt}  & 31.3 &  3\,406  & 28.9 & 11.52 & 78.4 &  3\,406  & 12.52 &  39\%  &  2.51\\
        \texttt{groebzig}  & 16.4 &  2\,132  & 24.1 & 10.79 & 44.6 &  2\,132  & 12.34 &  41\%  &  2.73\\[2pt]
        \texttt{aldenhoven}  & 83.2 &  9\,366  & 26.2 & 13.04 & 218.2 &  9\,260  & 14.9 &  39\%  &  2.62\\
        \texttt{andernach}  & 122.7 &  8\,039  & 47.4 & 13.72 & 293.6 &  8\,009  & 13.72 &  35\%  &  2.39\\
        \texttt{bad\_harzburg}  & 63.6 &  4\,929  & 41.2 & 13.29 & 151.3 &  4\,927  & 12.76 &  37\%  &  2.38\\
        \texttt{bad\_neunahr}  & 106.3 &  9\,068  & 34 & 11.47 & 268 &  9\,024  & 13.8 &  38\%  &  2.52\\
        \texttt{bergedorf}  & 239.8 &  20\,121  & 34.1 & 12.08 & 559.8 &  19\,902  & 15.35 &  36\%  &  2.33\\
        \texttt{celle}  & 359.3 &  26\,196  & 36.7 & 10.49 & 870.9 &  25\,964  & 15.56 &  37\%  &  2.42\\
        \texttt{euskirchen}  & 360.5 &  31\,434  & 29.7 & 11.1 & 880.7 &  30\,782  & 16.47 &  38\%  &  2.44\\
        \texttt{forsbach}  & 68 &  7\,609  & 27.1 & 13.85 & 173.6 &  7\,598  & 14.12 &  38\%  &  2.55\\[2pt]
        \texttt{gruppe1}  & 147.7 &  8\,788  & 52.9 & 12.61 & 289.4 &  8\,784  & 13.13 &  36\%  &  1.96\\
        \texttt{gruppe2}  & 65.9 &  6\,527  & 31.8 & 13.83 & 164.5 &  6\,507  & 13.75 &  38\%  &  2.50\\
        \texttt{gruppe3}  & 81.7 &  6\,984  & 35.6 & 12.92 & 209.9 &  6\,979  & 13.7 &  39\%  &  2.57\\
        \texttt{gruppe4}  & 121.6 &  10\,376  & 33 & 10.71 & 296.5 &  10\,356  & 14.16 &  39\%  &  2.44\\
        \texttt{gruppe5}  & 129 &  16\,317  & 20.4 & 9.84 & 360.3 &  16\,273  & 15.86 &  41\%  &  2.79\\
        \texttt{gruppe6}  & 134.9 &  12\,772  & 29.3 & 10.48 & 336.2 &  12\,761  & 14.59 &  39\%  &  2.49\\
        \texttt{gruppe7}  & 187.9 &  15\,793  & 34.7 & 11.64 & 415.8 &  15\,780  & 14.74 &  39\%  &  2.21\\
        \texttt{gruppe8}  & 54.2 &  6\,153  & 25.1 & 8.66 & 151.2 &  6\,151  & 13.54 &  40\%  &  2.79\\
        \texttt{gruppe9}  & 114.1 &  9\,605  & 35.9 & 12.1 & 275.1 &  9\,586  & 14.01 &  38\%  &  2.41\\
        \texttt{gruppe10}  & 283.2 &  25\,685  & 28.5 & 10.29 & 712.6 &  25\,295  & 16.01 &  40\%  &  2.52\\
        \texttt{gruppe11}  & 214.8 &  20\,287  & 27.2 & 8.52 & 550.4 &  20\,039  & 15.75 &  41\%  &  2.56\\
        \texttt{gruppe12}  & 304.2 &  20\,598  & 42.7 & 11.69 & 606.3 &  20\,486  & 14.82 &  38\%  &  1.99\\[2pt]
        \texttt{bonn}  &  2\,997.0  &  150\,943  & 28.1 & 13.8 &  5\,391.8  &  150\,827  & 18.18 &  41\%  &  1.80\\
        \texttt{cologne}  &  7\,917.8  &  330\,720  & 27.7 & 16.36 &  12\,946.9  &  329\,911  & 19.38 &  41\%  &  1.64\\
        \texttt{big\_set}  &  36\,452.3  &  1\,134\,906  & 30.7 & 27.14 &  46\,550.6  &  1\,089\,047  & 21.09 &  41\%  &  1.28\\
        \bottomrule
    \end{tabular*}
    } 
\end{table*}

\begin{table*}
    \center
    \caption{%
        Performance of \PBFS and \dichotomicscheme-IBFS on image segmentation instances. See Table~\ref{tbl:runtime-full-agg} for column descriptions.
    }%
    \label{tbl:runtime-full-seg}%
    \small{
    \begin{tabular*}{\textwidth}{@{\,}l@{\extracolsep{\fill}}r@{\extracolsep{\fill}}r@{\extracolsep{\fill}}r@{\extracolsep{\fill}}r@{\extracolsep{\fill}}r@{\extracolsep{\fill}}r@{\extracolsep{\fill}}r@{\extracolsep{\fill}}r@{\extracolsep{\fill}}r@{\,}}
        \toprule
        \multirow{2}{*}{Instance} & \multicolumn{4}{c}{\PBFS} & \multicolumn{4}{c}{Dichotomic with IBFS} & \multirow{2}{*}{Speedup} \\
        \cmidrule{2-5} \cmidrule{6-9}
        & Time & BP & Ad./BP & Loop/Init & Time & BP & Vert. & Contr. \\
        \midrule
        \texttt{adhead.n6c10}  &  4\,037.2  & 2 & 2.5 & 0 &  10\,150.9  & 3 & 0.05 &  50\%  &  2.51\\
        \texttt{adhead.n6c100}  &  4\,053.3  & 2 & 2.5 & 0 &  10\,248.2  & 3 & 0.05 &  50\%  &  2.53\\
        \texttt{adhead.n26c10}  &  17\,671.7  &  10\,311  & 209 & 0.18 &  37\,421.2  &  5\,690  & 0.26 &  46\%  &  2.12\\
        \texttt{adhead.n26c100}  &  19\,307.7  &  36\,998  & 52.5 & 0.19 &  38\,552.9  &  31\,771  & 0.36 &  44\%  &  2.00\\[3pt]
        \texttt{babyface.n6c10}  &  1\,524.7  & 20 & 29.3 & 0 &  3\,892.4  & 21 & 0 &  49\%  &  2.55\\
        \texttt{babyface.n6c100}  &  1\,475.4  & 29 & 22.3 & 0 &  3\,670.3  & 30 & 0 &  50\%  &  2.49\\
        \texttt{babyface.n26c10}  &  3\,754.8  & 277 & 99.9 & 0.01 &  10\,700.9  & 253 & 0.01 &  53\%  &  2.85\\
        \texttt{babyface.n26c100}  &  3\,829.0  & 336 & 80.8 & 0.01 &  10\,410.5  & 336 & 0.01 &  53\%  &  2.72\\[3pt]
        \texttt{bone.n6c10}  &  2\,524.4  & 84 & 311.8 & 0.01 &  6\,009.4  & 82 & 0.03 &  48\%  &  2.38\\
        \texttt{bone.n6c100}  &  2\,561.9  & 124 & 262.2 & 0.01 &  6\,008.9  & 125 & 0.03 &  48\%  &  2.35\\
        \texttt{bone.n26c10}  &  14\,172.4  &  6\,225  &  1\,313.1  & 0.92 &  23\,458.2  &  3\,686  & 0.28 &  43\%  &  1.66\\
        \texttt{bone.n26c100}  &  15\,456.7  &  12\,189  & 715.9 & 0.94 &  25\,101.5  &  11\,204  & 0.34 &  42\%  &  1.62\\[3pt]
        \texttt{liver.n6c10}  &  1\,365.6  & 185 & 82.4 & 0.01 &  3\,332.9  & 181 & 0.08 &  48\%  &  2.44\\
        \texttt{liver.n6c100}  &  1\,429.3  &  1\,078  & 13.9 & 0.01 &  3\,427.7  &  1\,077  & 0.09 &  47\%  &  2.40\\
        \texttt{liver.n26c10}  &  3\,054.8  & 378 & 100.9 & 0.01 &  9\,054.2  & 355 & 0.09 &  55\%  &  2.96\\
        \texttt{liver.n26c100}  &  3\,127.8  &  1\,506  & 25 & 0.01 &  8\,865.7  &  1\,493  & 0.1 &  54\%  &  2.83\\
        \bottomrule
    \end{tabular*}
    } 
\end{table*}
Tables~\ref{tbl:runtime-full-agg} and~\ref{tbl:runtime-full-seg} report the performance of \dichotomicscheme with IBFS as the max-flow algorithm (denoted as \dichotomicscheme-IBFS) and \PBFS.
We choose IBFS because it outperforms PRF on image segmentation instances, which are grid-like and have short~$\source$-$\sink$-paths~\cite{GHKTW11}.
The comparison with \dichotomicscheme-PRF in Appendix~\ref{app:experiments} confirms that this advantage carries over to the aggregation instances, which are nearly planar and also have short~$\source$-$\sink$-paths.
On both instance classes, \PBFS is significantly faster than \dichotomicscheme-IBFS.
The speedup is between~2 and~3 on most instances and somewhat smaller on the largest aggregation instances and the \texttt{bone.n26} instances.

There are significant differences between the two instance classes:
The aggregation instances have around~$n/3$ breakpoints and the amount of work performed by \PBFS per breakpoint (measured by the number of adoptions per breakpoint) is consistently small.
By contrast, the segmentation instances have fewer breakpoints.
On some instances, the amount of work per breakpoint is higher, which indicates that the sink-component tree changes more drastically.
Another difference is the impact of the initial flow computation.
On the aggregation instances, the solution for~$\maxLambda$ consists only of the input polygons, whereas the solution for~$\minLambda$ typically contains most triangles.
Hence, the parts of the graph that are contracted by \dichotomicscheme before the first~$\Bisect$ call are very small.
Similarly, \PBFS computes the initial flow quickly and spends most of the runtime in the main loop.
The segmentation instances have more balanced initial cuts, so the initial flow computation is the bottleneck for both algorithms.
\PBFS has an advantage here because it only makes one max-flow call, whereas \dichotomicscheme makes two and also needs to contract the graph.
This explains the reduced speedup on the \texttt{bone.n26} instances, where the contracted graphs are relatively large and the work per breakpoint is high.

Note that the number of breakpoints computed by both algorithms is not identical.
The differences are due to numerical instability, which affects both algorithms because they perform divisions between floating-point numbers.
Because the algorithms explore the parameter space in different orders, this can lead to small disagreements.
However, we verified that for each breakpoint of either algorithm, the computed flow values are identical up to the sixth decimal point.

We observe that GGT cannot realistically outperform \dichotomicscheme on any of the tested instances.
On the segmentation instances, the first contraction step already removes most vertices.
On the aggregation instances, the average number of contracted graphs that contain a given vertex (see Vert in Tables~\ref{tbl:runtime-full-agg} and~\ref{tbl:runtime-full-seg}) is close to~$\log_2 n$, which indicates that the splits are very balanced.
Note that for a worst-case instance, where every~\Bisect call splits off one vertex, this number would be linear in~$n$.

\subparagraph*{Influence of Instance Structure.}

\begin{table}[t]
    \center
    \caption{%
        Performance on \texttt{liver.n6c10} depending on the number of parametric source-/sink-incident edges. Par is the proportion of vertices for which the source-incident ($\source$) or sink-incident ($\sink$) edge is given a randomly chosen parametric capacity. The remaining source- or sink-incident edges are given capacity~0. An exception is~$t = 0.0$, in which case the sink-incident edges retain their original, constant capacities to ensure that the minimum cut is not trivial. Edges that are not incident to~$\source$ or~$\sink$ also keep their original, constant capacities. Bot/BP is the number of \emph{bottleneck edges}, i.e., tree edges that are removed in \texttt{UpdateResidualGraph}, divided by the number of breakpoints. Ad/Bot is the number of successful adoptions, divided by the number of bottleneck edges. Dist is the average distance label among successfully adopted orphans. Other columns as in Table~\ref{tbl:runtime-full-agg}, except times are given in seconds.
    }%
    \label{tbl:liver}%
    \small{
    \begin{tabular*}{\textwidth}{@{\,}r@{\extracolsep{\fill}}r@{\extracolsep{\fill}}r@{\extracolsep{\fill}}r@{\extracolsep{\fill}}r@{\extracolsep{\fill}}r@{\extracolsep{\fill}}r@{\extracolsep{\fill}}r@{\extracolsep{\fill}}r@{\extracolsep{\fill}}r@{\extracolsep{\fill}}r@{\extracolsep{\fill}}r@{\extracolsep{\fill}}r@{\,}}
        \toprule
        \multicolumn{2}{c}{Par} & \multicolumn{6}{c}{\PBFS} & \multicolumn{4}{c}{\dichotomicscheme-IBFS} & \multirow{2}{*}{Speedup} \\
        \cmidrule{1-2} \cmidrule{3-8} \cmidrule{9-12}
        $\source$ & $\sink$ & Time & BP & Bot/BP & Ad/Bot & Dist & Loop/Init & Time & BP & Vert & Contr \\
        \midrule
        0.1 & 0.0 & 15.1 & 1\,533 & 132.8 & 119.7 & 87.9 & 1.64 & 15.8 & 937 & 0.9 & 19\% & 1.05 \\
        0.1 & 0.1 & 258.2 & 1\,548 & 488.3 & 554.2 & 101.0 &  72.91 & 156.7 & 679 & 5.1 & 10\% & 0.61 \\
        0.1 & 0.5 & 5.8 & 464 & 1\,471.4 & 5.5 & 57.9 &  1.41 & 5.6 & 412 & 0.1 & 28\% & 0.97\\
        0.1 & 1.0 & 3.8 & 418 & 1\,499.2 & 3.5 & 313.1 & 1.29 & 4.2 & 379 & 0.1 & 38\% & 1.10 \\[3pt]
        0.5 & 0.0 & 1.9 & 511 & 18.9 & 18.4 & 17.4 &  0.05 & 3.9 & 438 & 0.1 & 42\% & 2.06 \\
        0.5 & 0.1 & 2.3 & 388 & 64.5 & 5.0 & 7.5 &  0.04 & 4.8 & 421 & 0.1 & 34\% & 2.07 \\
        0.5 & 0.5 & 379.9 & 8\,192 & 356.9 & 171.0 & 101.6 &  113.93 & 236.4 & 3\,506 & 6.5 & 8\% & 0.62 \\
        0.5 & 1.0 & 32.9 & 4\,566 & 941.5 & 5.4 & 38.3 &  15.92 & 34.4 & 3\,093 & 1.2 & 9\% & 1.05 \\[3pt]
        1.0 & 0.0 & 1.4 & 185 & 10.6 & 7.8 & 11.7 & 0.01 & 3.3 & 181 & 0.1 & 48\% & 2.34 \\
        1.0 & 0.1 & 1.8 & 361 & 63.4 & 2.8 & 15.6 & 0.03 & 3.9 & 384 & 0.1 & 41\% & 2.23 \\
        1.0 & 0.5 & 33.6 & 3\,748 & 53.9 & 90.9 & 133.2 & 0.30 & 36.1 & 3\,066 & 1.2 & 9\% & 1.08 \\
        1.0 & 1.0 & 482.3 & 10\,330 & 428.4 & 141.6 & 142.2 & 206.61 & 218.2 & 4\,451 & 6.9 & 9\% & 0.45 \\
        \bottomrule
    \end{tabular*}}
\end{table}
Because we do not have access to realistic parametric capacities for the image segmentation instances, we investigate how the choice of the capacities affects the performance of the algorithms.
Table~\ref{tbl:liver} shows the results of an experiment on the \texttt{liver.n6c10} instance, in which we vary the amount of source- and sink-incident edges with parametric capacities.
Note that these are the only tested instances in which both the source- and sink-incident edges are parametric.
In polygon aggregation as well as the computer vision problems that can be represented in the MAP-MRF framework~\cite{KBR07}, the sink-incident edges have constant capacities.

In instances with many parametric-sink incident edges, \PBFS removes a large number of saturated tree edges per breakpoint.
This is because all residual sink-incident edges are tree edges.
If they are parametric, they are likely to become saturated at some point, which may force a significant portion of the tree to be rebuilt.
If the proportions of parametric source- and sink-incident edges are similar, we observe that the number of adoptions per removed tree edge becomes very high.
This indicates that \PBFS performs many iterations that do not add a breakpoint but drastically change the sink-component tree.
In this case, \dichotomicscheme is faster because the number of max-flow computations is proportional to the number of breakpoints.

Overall, \PBFS is faster on instances where the residual graph does not change drastically between breakpoints.
This is the case for both polygon aggregation instances and the evaluated image segmentation instances in which the parametric source-incident edges significantly outnumber the parametric sink-incident edges.
On instances with few breakpoints and drastic changes in the residual graph, \dichotomicscheme performs better.

\subparagraph*{Approximation.}
\begin{figure}[t]
    \newcommand{\plotHeight}{4.5cm}
\newcommand{\plotWidth}{0.49\textwidth}

\begin{tikzpicture}
\begin{axis}[
   name=mainPlot,
   height=\plotHeight,
   width=\plotWidth,
   xmin=-0.25,
   xmax=16.25,
   ymin=-0.03,
   ymax=1.03,
   xlabel={Approximation factor ($\varepsilon = 10^{-x}$)},
   ylabel={Found breakpoints},
   xlabel style = {yshift=8pt},
   ylabel style = {yshift=-1pt},
   xtick = {0, 2, 4, 6, 8, 10, 12, 14},
   extra x ticks = {16},
   extra x tick labels = { $\infty$ },
   ytick = {0, 0.25, 0.5, 0.75, 1},
   yticklabels = {0\%, 25\%, 50\%, 75\%, 100\%},
   xtick pos=left,
   ytick pos=left,
   ytick align=outside,
   xtick align=outside,
   grid=both,
   axis line style={legendColor},
   at={(0.0\linewidth,0)}
]
\addplot[color=plotColor1,linePlot,mark=*] table[y=breakpointsAhrem,col sep=tab]{plot.csv};
\addplot[color=plotColor2,linePlot,mark=*] table[y=breakpointsEuskirchen,col sep=tab]{plot.csv};
\addplot[color=plotColor3,linePlot,mark=*] table[y=breakpointsBonn,col sep=tab]{plot.csv};
\addplot[color=plotColor4,linePlot,mark=*] table[y=breakpointsCologne,col sep=tab]{plot.csv};
\addplot[color=plotColor5,linePlot,mark=*] table[y=breakpointsBigSet,col sep=tab]{plot.csv};
\end{axis}

\begin{axis}[
   scaled y ticks = false,
   y tick label style={/pgf/number format/fixed},
   height=\plotHeight,
   width=\plotWidth,
   xmin=-0.25,
   xmax=16.25,
   ymin=0,
   ymax=4,
   xlabel={Approximation factor ($\varepsilon = 10^{-x}$)},
   ylabel={Speedup},
   xlabel style = {yshift=8pt},
   ylabel style = {yshift=-1pt},
   xtick = {0, 2, 4, 6, 8, 10, 12, 14},
   extra x ticks = {16},
   extra x tick labels = { $\infty$ },
   xtick pos=left,
   ytick pos=left,
   ytick align=outside,
   xtick align=outside,
   grid=both,
   axis line style={legendColor},
   at={(0.5\linewidth,0)}
]

\addplot[linePlot,mark=none,domain=-0.25:16.25] {1};
\addplot[color=plotColor1,linePlot,mark=*] table[y=slowdownAhrem,col sep=tab]{plot.csv};\label{legend:ahrem}
\addplot[color=plotColor2,linePlot,mark=*] table[y=slowdownEuskirchen,col sep=tab]{plot.csv};\label{legend:euskirchen}
\addplot[color=plotColor3,linePlot,mark=*] table[y=slowdownBonn,col sep=tab]{plot.csv};\label{legend:bonn}
\addplot[color=plotColor4,linePlot,mark=*] table[y=slowdownCologne,col sep=tab]{plot.csv};\label{legend:cologne}
\addplot[color=plotColor5,linePlot,mark=*] table[y=slowdownBigSet,col sep=tab]{plot.csv};\label{legend:bigSet}
\end{axis}

\arrayrulecolor{legendColor}
\node[inner sep=0pt,outer sep=0pt,yshift=-3pt] (legend) at (mainPlot.outer south west) {};
\node[inner sep=0pt,outer sep=0pt,anchor=north west] at (legend) {
\footnotesize
\begin{tabular*}{\textwidth}{|@{~}l@{~~~~}r@{\extracolsep{\fill}}r@{\extracolsep{\fill}}r@{\extracolsep{\fill}}r@{\extracolsep{\fill}}r@{~}|}
    \hline
     & & & & & \\[-6pt]
    \textbf{Instance:} & \legend{\ref{legend:ahrem}}\,\texttt{ahrem} & \legend{\ref{legend:euskirchen}}\,\texttt{euskirchen} & \legend{\ref{legend:bonn}}\,\texttt{bonn} & \legend{\ref{legend:cologne}}\,\texttt{cologne} & \legend{\ref{legend:bigSet}}\,\texttt{big\_set} \\[1pt]
    \hline
\end{tabular*}
};
\arrayrulecolor{black}
\end{tikzpicture}%
    \caption{%
    	Performance of \dichotomicscheme-IBFS on aggregation instances, depending on the approximation factor~$\varepsilon=10^{-x}$.
        Left: Ratio of found breakpoints, compared to~$\varepsilon=0$. Right: Speedup of \PBFS over \dichotomicscheme-IBFS.
    }%
    \label{fig:plot}%
\end{figure}
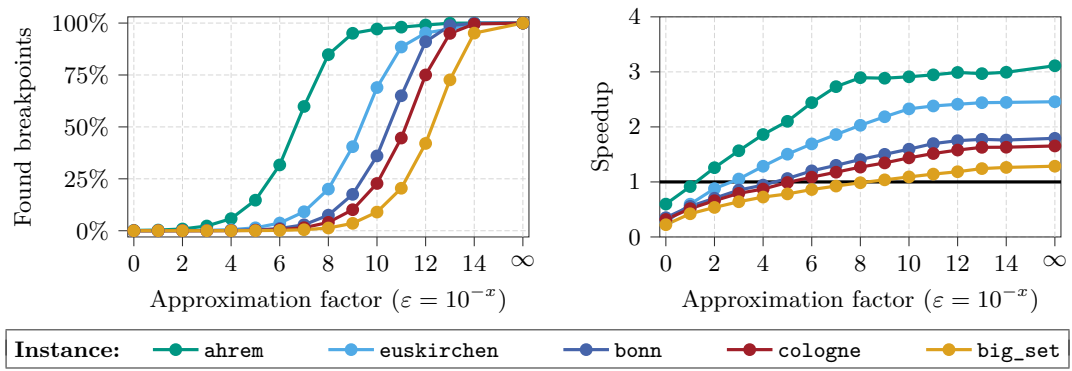
Although our main focus is on computing the entire breakpoint function, for some applications it is sufficient to compute an~$\varepsilon$-approximation.
As discussed in Section~\ref{sec:applications}, this can be done easily with \dichotomicscheme.
In this use case, \dichotomicscheme potentially has a runtime advantage over \PBFS, which always computes all breakpoints.
However, Figure~\ref{fig:plot} shows that on aggregation instances, \PBFS starts to outperform \dichotomicscheme-IBFS for~$\varepsilon$ values between~$10^{-2}$ and~$10^{-8}$.
For these values, \dichotomicscheme only finds a tiny fraction of all breakpoints.
If~$\varepsilon$ is chosen small enough to find half of all breakpoints, then \PBFS outperforms \dichotomicscheme by a factor between~1.2 and~2.5.

\subparagraph*{Other Algorithms.}
Rottmann et al.~\cite{RDGRH21} solve polygon aggregation using \dichotomicscheme without contraction.
Even on moderately sized instances, this approach does not scale well enough to compute the entire breakpoint function in reasonable time.
With an approximation factor of~$\varepsilon=10^{-7}$ on \texttt{euskirchen}, the algorithm runs for~145\,s and computes~2\,815 breakpoints (less than~9\%).
On \texttt{bonn} with~$\varepsilon=10^{-5}$, it takes~254\,s and computes~560 breakpoints (less than~0.4\%).

Finally, \PBFS significantly outperforms the algorithm for the simple sensitivity analysis proposed by Gallo, Grigoriadis and Tarjan~\cite{GGT89}.
We implemented this approach using EIBFS as the static max-flow algorithm, which is faster than PRF on our instances, and ran it for the set of breakpoints computed by \PBFS.
We observe that the algorithm does not scale well on aggregation instances:
On \texttt{ahrem}, it takes~257.0\,ms, which is over~20 times slower than \PBFS.
On \texttt{euskirchen}, the runtime is~71\,s, which is almost~200 times slower.
The reason for this is that the effort for turning a flow into a preflow is proportional to the number of source- and sink-incident edges, and this step has to be performed once for every breakpoint.
On aggregation instances, this requires quadratic effort overall.

\section{Conclusion}
We showed that the polygon aggregation problem posed by Rottmann et al.~\cite{RDGRH21} can be formulated as a parametric min-cut problem with monotone source-sink-capacities~(\MPMC), which allows faster algorithms to be applied.
With this application in mind, we presented parametric breadth-first search~(\PBFS), a new algorithm for \MPMC that finds the breakpoints in order.
\PBFS outperforms the state of the art on most evaluated benchmark instances, in particular the aggregation instances, while bounding the worst-case runtime to that of a single static max-flow call.

\subparagraph*{Open Questions.}
Our new \PBFS algorithm reuses some of the subroutines from the incremental breadth-first search~(IBFS) algorithm for the static max-flow algorithm~\cite{GHKTW11}.
It would be interesting to investigate whether the approach of finding the breakpoints in order can be combined with other max-flow algorithms, such as Hochbaum's pseudoflow algorithm~\cite{Hoc08}.
Furthermore, there is room for improvement in the runtime bound of~$\mathcal O(n^2m)$.
IBFS can be reduced to~$\mathcal O(nm \log n)$ with dynamic trees, but as discussed in Section~\ref{sec:alg}, it is unclear whether they can be adapted to support the operation of finding the smallest root of the residual capacity functions.

\bibliography{literature}

\appendix
\section{Omitted Experiments}
\label{app:experiments}
Tables~\ref{tbl:prf-agg} and~\ref{tbl:prf-seg} compare the performance of \dichotomicscheme-PRF to \dichotomicscheme-IBFS.
The latter is faster on all instances.
The speedup ranges from~25--50\% on the small aggregation instances to almost~9 on some segmentation instances.

\begin{table*}
    \center
    \caption{%
        Runtimes of \PBFS, \dichotomicscheme-IBFS and \dichotomicscheme-PRF on polygon aggregation instances. BP is the number of breakpoints computed by \PBFS.
    }%
    \label{tbl:prf-agg}%
    \small{
    \begin{tabular*}{\textwidth}{@{\,}l@{\extracolsep{\fill}}r@{\extracolsep{\fill}}r@{\extracolsep{\fill}}r@{\extracolsep{\fill}}r@{\extracolsep{\fill}}r@{\extracolsep{\fill}}r@{\extracolsep{\fill}}r@{\extracolsep{\fill}}r@{\,}}
        \toprule
        \multirow{3}{*}{Instance} & \multirow{3}{*}{BP} & \multicolumn{3}{c}{Runtime [ms]} & \multicolumn{3}{c}{Speedup} \\
        \cmidrule{3-5} \cmidrule{6-8}
        & & \multirow{2}{*}{\PBFS} & \dichotomicscheme- & \dichotomicscheme- & \PBFS & \PBFS & IBFS \\
        & & & IBFS & PRF & /IBFS & /PRF & /PRF \\
        \midrule
        \texttt{ahrem}  &  1\,891  & 11.9 & 35.1 & 44.3 & 2.95 & 3.72 &  1.26\\
        \texttt{altenholz}  &  3\,342  & 48.4 & 109.9 & 185.7 & 2.27 & 3.83 &  1.69\\
        \texttt{altenrath}  &  2\,742  & 22.1 & 59.2 & 84.4 & 2.68 & 3.82 &  1.43\\
        \texttt{belm}  &  2\,421  & 36.9 & 85.3 & 130.7 & 2.31 & 3.54 &  1.53\\
        \texttt{berga}  &  1\,977  & 16 & 43.5 & 57.4 & 2.71 & 3.58 &  1.32\\
        \texttt{bockelskamp}  & 859 & 6.3 & 16.7 & 21 & 2.64 & 3.31 &  1.26\\
        \texttt{bokeloh}  &  1\,036  & 11.5 & 27.9 & 38.1 & 2.43 & 3.32 &  1.36\\
        \texttt{braunlage}  &  2\,383  & 28.9 & 69 & 99.2 & 2.39 & 3.43 &  1.44\\
        \texttt{buisdorf}  &  2\,991  & 27 & 68.5 & 96.1 & 2.54 & 3.56 &  1.40\\
        \texttt{butzweiler}  &  1\,347  & 12.6 & 29.2 & 38.4 & 2.31 & 3.05 &  1.32\\
        \texttt{duengenheim}  &  2\,581  & 18 & 50 & 65 & 2.78 & 3.62 &  1.30\\
        \texttt{edendorf}  &  4\,953  & 51.4 & 124.6 & 180.9 & 2.42 & 3.52 &  1.45\\
        \texttt{erlenbach}  &  1\,624  & 14.1 & 37.2 & 51.7 & 2.63 & 3.65 &  1.39\\
        \texttt{erp}  &  4\,663  & 32.8 & 91.4 & 124.8 & 2.79 & 3.81 &  1.36\\
        \texttt{friesheim}  &  2\,387  & 22.5 & 55.1 & 75.6 & 2.45 & 3.36 &  1.37\\
        \texttt{gerolstein}  &  2\,577  & 31.4 & 73.2 & 102.4 & 2.33 & 3.26 &  1.40\\
        \texttt{gevenich}  &  1\,426  & 10.5 & 28.6 & 36.1 & 2.72 & 3.44 &  1.26\\
        \texttt{gluecksburg}  &  2\,741  & 26.2 & 65.7 & 90.7 & 2.51 & 3.47 &  1.38\\
        \texttt{goddula}  &  1\,758  & 15 & 39.5 & 53.2 & 2.64 & 3.55 &  1.35\\
        \texttt{goldenstedt}  &  3\,406  & 31.3 & 78.4 & 110.6 & 2.51 & 3.54 &  1.41\\
        \texttt{groebzig}  &  2\,132  & 16.4 & 44.6 & 58.7 & 2.73 & 3.59 &  1.32\\[3pt]
        \texttt{aldenhoven}  &  9\,366  & 83.2 & 218.2 & 317.3 & 2.62 & 3.81 &  1.45\\
        \texttt{andernach}  &  8\,039  & 122.7 & 293.6 & 474.7 & 2.39 & 3.87 &  1.62\\
        \texttt{bad\_harzburg}  &  4\,929  & 63.6 & 151.3 & 237 & 2.38 & 3.73 &  1.57\\
        \texttt{bad\_neunahr}  &  9\,068  & 106.3 & 268 & 428.2 & 2.52 & 4.03 &  1.60\\
        \texttt{bergedorf}  &  20\,121  & 239.8 & 559.8 & 937.8 & 2.33 & 3.91 &  1.68\\
        \texttt{celle}  &  26\,196  & 359.3 & 870.9 &  1\,526.9  & 2.42 & 4.25 &  1.75\\
        \texttt{euskirchen}  &  31\,434  & 360.5 & 880.7 &  1\,438.1  & 2.44 & 3.99 &  1.63\\
        \texttt{forsbach}  &  7\,609  & 68 & 173.6 & 261.8 & 2.55 & 3.85 &  1.51\\[3pt]
        \texttt{gruppe1}  &  8\,788  & 147.7 & 289.4 & 425.8 & 1.96 & 2.88 &  1.47\\
        \texttt{gruppe2}  &  6\,527  & 65.9 & 164.5 & 243.8 & 2.5 & 3.7 &  1.48\\
        \texttt{gruppe3}  &  6\,984  & 81.7 & 209.9 & 326.3 & 2.57 & 3.99 &  1.55\\
        \texttt{gruppe4}  &  10\,376  & 121.6 & 296.5 & 436.2 & 2.44 & 3.59 &  1.47\\
        \texttt{gruppe5}  &  16\,317  & 129 & 360.3 & 510.7 & 2.79 & 3.96 &  1.42\\
        \texttt{gruppe6}  &  12\,772  & 134.9 & 336.2 & 505 & 2.49 & 3.74 &  1.50\\
        \texttt{gruppe7}  &  15\,793  & 187.9 & 415.8 & 597.5 & 2.21 & 3.18 &  1.44\\
        \texttt{gruppe8}  &  6\,153  & 54.2 & 151.2 & 206.8 & 2.79 & 3.81 &  1.37\\
        \texttt{gruppe9}  &  9\,605  & 114.1 & 275.1 & 424.3 & 2.41 & 3.72 &  1.54\\
        \texttt{gruppe10}  &  25\,685  & 283.2 & 712.6 &  1\,104.8  & 2.52 & 3.9 &  1.55\\
        \texttt{gruppe11}  &  20\,287  & 214.8 & 550.4 & 776.6 & 2.56 & 3.62 &  1.41\\
        \texttt{gruppe12}  &  20\,598  & 304.2 & 606.3 & 936.6 & 1.99 & 3.08 &  1.54\\[3pt]
        \texttt{bonn}  &  150\,943  &  2\,997.0  &  5\,391.8  &  9\,885.5  & 1.8 & 3.3 &  1.83\\
        \texttt{cologne}  &  330\,720  &  7\,917.8  &  12\,946.9  &  23\,904.4  & 1.64 & 3.02 &  1.85\\
        \texttt{big\_set}  &  1\,134\,906  &  36\,452.3  &  46\,550.6  &  102\,488.0  & 1.28 & 2.81 &  2.20\\
        \bottomrule
    \end{tabular*}
    } 
\end{table*}

\begin{table*}
    \center
    \caption{%
        Runtimes of \PBFS, \dichotomicscheme-IBFS and \dichotomicscheme-PRF on image segmentation instances. BP is the number of breakpoints computed by \PBFS.
    }%
    \label{tbl:prf-seg}%
    \small{
    \begin{tabular*}{\textwidth}{@{\,}l@{\extracolsep{\fill}}r@{\extracolsep{\fill}}r@{\extracolsep{\fill}}r@{\extracolsep{\fill}}r@{\extracolsep{\fill}}r@{\extracolsep{\fill}}r@{\extracolsep{\fill}}r@{\extracolsep{\fill}}r@{\,}}
        \toprule
        \multirow{3}{*}{Instance} & \multirow{3}{*}{BP} & \multicolumn{3}{c}{Runtime [ms]} & \multicolumn{3}{c}{Speedup} \\
        \cmidrule{3-5} \cmidrule{6-8}
        & & \multirow{2}{*}{\PBFS} & \dichotomicscheme- & \dichotomicscheme- & \PBFS & \PBFS & IBFS \\
        & & & IBFS & PRF & /IBFS & /PRF & /PRF \\
        \midrule
        \texttt{adhead.n6c10}  & 2 &  4\,037.2  &  10\,150.9  &  25\,828.5  & 2.51 & 6.4 &  2.54\\
        \texttt{adhead.n6c100}  & 2 &  4\,053.3  &  10\,248.2  &  33\,759.6  & 2.53 & 8.33 &  3.29\\
        \texttt{adhead.n26c10}  &  10\,311  &  17\,671.7  &  37\,421.2  &  147\,513.4  & 2.12 & 8.35 &  3.94\\
        \texttt{adhead.n26c100}  &  36\,998  &  19\,307.7  &  38\,552.9  &  134\,413.0  & 2 & 6.96 &  3.49\\[3pt]
        \texttt{babyface.n6c10}  & 20 &  1\,524.7  &  3\,892.4  &  10\,878.6  & 2.55 & 7.13 &  2.79\\
        \texttt{babyface.n6c100}  & 29 &  1\,475.4  &  3\,670.3  &  14\,464.3  & 2.49 & 9.8 &  3.94\\
        \texttt{babyface.n26c10}  & 277 &  3\,754.8  &  10\,700.9  &  79\,058.7  & 2.85 & 21.06 &  7.39\\
        \texttt{babyface.n26c100}  & 336 &  3\,829.0  &  10\,410.5  &  92\,144.7  & 2.72 & 24.06 &  8.85\\[3pt]
        \texttt{bone.n6c10}  & 84 &  2\,524.4  &  6\,009.4  &  23\,820.2  & 2.38 & 9.44 &  3.96\\
        \texttt{bone.n6c100}  & 124 &  2\,561.9  &  6\,008.9  &  24\,511.9  & 2.35 & 9.57 &  4.08\\
        \texttt{bone.n26c10}  &  6\,225  &  14\,172.4  &  23\,458.2  &  81\,688.2  & 1.66 & 5.76 &  3.48\\
        \texttt{bone.n26c100}  &  12\,189  &  15\,456.7  &  25\,101.5  &  87\,948.6  & 1.62 & 5.69 &  3.50\\[3pt]
        \texttt{liver.n6c10}  & 185 &  1\,365.6  &  3\,332.9  &  10\,636.7  & 2.44 & 7.79 &  3.19\\
        \texttt{liver.n6c100}  &  1\,078  &  1\,429.3  &  3\,427.7  &  11\,700.3  & 2.4 & 8.19 &  3.41\\
        \texttt{liver.n26c10}  & 378 &  3\,054.8  &  9\,054.2  &  21\,842.0  & 2.96 & 7.15 &  2.41\\
        \texttt{liver.n26c100}  &  1\,506  &  3\,127.8  &  8\,865.7  &  25\,896.4  & 2.83 & 8.28 &  2.92\\
        \bottomrule
    \end{tabular*}
    } 
\end{table*}

\end{document}